\documentclass{article}

\usepackage[margin=1.2in]{geometry}

\usepackage[utf8]{inputenc} 
\usepackage[T1]{fontenc}    
\usepackage{hyperref}       
\usepackage{url}             
\usepackage{booktabs}       
\usepackage{amsfonts}       
\usepackage{nicefrac}       
\usepackage{microtype}      
\usepackage{lipsum}
\usepackage[round]{natbib}
\usepackage{wrapfig}
\usepackage{tikz}
\usetikzlibrary{shapes,arrows}
\usepackage{booktabs,threeparttable}

\newcommand{\newlab}[2]{#2\def\@currentlabel{#2}\label{#1}}
\usepackage{array}

\newcolumntype{L}[1]{>{\raggedleft\let\newline\\\arraybackslash\hspace{0pt}}m{#1}}
\newcolumntype{H}{>{\setbox0=\hbox\bgroup}c<{\egroup}@{}}
\usepackage{xspace}
\newcommand{\gmm}{$\text{AdaPT-GMM}_g$\xspace}
\newcommand{\FDR}{\text{FDR}}

\newcommand{\adaptg}{$\text{AdaPT}_g$\xspace}
\newcommand{\cM}{\mathcal{M}}
\newcommand{\cN}{\mathcal{N}}
\newcommand{\cH}{\mathcal{H}}
\newcommand{\setcomp}{\mathsf{c}}

\newcommand{\hq}{\hat{q}}
\newcommand{\heta}{\hat{\eta}}
\newcommand{\hFDP}{\widehat{\text{FDP}}}
\newcommand{\lambdamin}{\lambda}
\newcommand{\lambdamax}{\nu}
\newcommand{\alphamax}{\alpha_m}

\newcommand{\E}{\mathbb{E}}

\newcommand{\sgn}{\text{sgn}}

\newcommand{\RR}{\mathbb{R}}
\newcommand{\PP}{\mathbb{P}}
\newcommand{\EE}{\mathbb{E}}

\newcommand{\cX}{\mathcal{X}}

\usepackage{amsmath,amsthm,amssymb}
\DeclareMathOperator*{\argmax}{arg\,max}

\newcommand{\simiid}{\overset{\textrm{i.i.d.}}{\sim}}
\newcommand{\simind}{\overset{\textrm{ind.}}{\sim}}
\usepackage{todonotes}
\usepackage{enumerate}
\usepackage{caption}
\usepackage{bbm}
\usepackage{subcaption}
\usepackage{algorithm}
\usepackage[noend]{algpseudocode}
\newtheorem{theorem}{Theorem}[section]
\newtheorem{lemma}[theorem]{Lemma}

\newcommand\numberthis{\addtocounter{equation}{1}\tag{\theequation}}

\usepackage{pdflscape}
\title{AdaPT-GMM: Powerful and robust covariate-assisted multiple testing}

\author{
  Patrick Chao\\
  \texttt{pchao@wharton.upenn.edu} \\
 William Fithian \\
  \texttt{wfithian@berkeley.edu} \\
}

\begin{document}
\maketitle

\begin{abstract}
We propose a new empirical Bayes method for covariate-assisted multiple testing with false discovery rate (FDR) control, where we model the local false discovery rate for each hypothesis as a function of both its covariates and $p$-value. Our method refines the adaptive $p$-value thresholding (AdaPT) procedure by generalizing its masking scheme to reduce the bias and variance of its false discovery proportion estimator, improving the power when the rejection set is small or some null $p$-values concentrate near 1. We also introduce a Gaussian mixture model for the conditional distribution of the test statistics given covariates, modeling the mixing proportions with a generic user-specified classifier, which we implement using a two-layer neural network.  Like AdaPT, our method provably controls the FDR in finite samples even if the classifier or the Gaussian mixture model is misspecified. We show in extensive simulations and real data examples that our new method, which we call \gmm, consistently delivers high power relative to competing state-of-the-art methods. In particular, it performs well in scenarios where AdaPT is underpowered, and is especially well-suited for testing composite null hypothesis, such as whether the effect size exceeds a practical significance threshold.

\end{abstract}
\section{Introduction}\label{sec: Introduction}

\subsection{Multiple Testing with Covariates}\label{subsec: Multiple Testing}

In most high-throughput multiple testing applications, the hypotheses are not {\em a priori} exchangeable, but rather have known histories and meaningful relationships to one another. For example, when screening many genetic point mutations for marginal association with a given phenotype, we typically have a great deal of prior side information about each mutation including estimated associations with other, related phenotypes; information from gene ontologies about what pathways its gene is involved in; the minor allele frequency, and more. 

We consider the problem of {\em covariate-assisted multiple testing} where for each null hypothesis $H_i,$ $i=1,\ldots,n$, we observe not only a $p$-value $p_i$ but also a {\em covariate} or {\em predictor} $x_i$ in some generic predictor space $\cX$, typically $\RR^d$. We assume throughout that the covariates are fixed, or equivalently that the distributional assumptions on the $p$-values hold after conditioning on $x_1,\ldots,x_n$. Our aim is to test the hypotheses while controlling at some prespecified significance level $\alpha$ the {\em false discovery rate} (FDR), defined as $\EE[V / \max\{R,1\}]$, where $R$ is the number of total rejections and $V$ is the number of rejected true null hypotheses (or ``false discoveries'') \citep{BH}. The random variable $V/\max\{R,1\}$ is called the {\em false discovery proportion} (FDP).

We are especially interested in the common setting where the $p$-values are derived from $z$-values $z_i \sim \cN(\theta_i, \sigma_i^2)$, with known standard error $\sigma_i > 0$. Although the most common null hypothesis to test in this setting is the {\em point null} $H_i:\; \theta_i = 0$, it can be more scientifically interesting to test the {\em one-sided null} $H_i:\; \theta_i \leq 0$ (or $H_i:\;\theta_i \geq 0$) or the {\em interval null} $H_i:\; |\theta_i| \leq \delta$ for some minimum effect size of interest $\delta > 0$. 

In Bayesian terms, our goal is to learn what the covariate $x_i$ tells us about the prior likelihood that $H_i$ is true and the distribution of $p_i$ under the null and alternative, which together determine the {\em local false discovery rate} (lfdr), the posterior likelihood that $H_i$ is true after observing $p_i$. For example, by analyzing the data jointly we might learn that $p_i=0.02$ indicates a promising lead when $x_i$ falls in a particular signal-rich region of the predictor space, but probably reflects noise when $x_i$ falls in a region with very few non-null hypotheses. If so, we can favor hypotheses in the first region by using a covariate-dependent {\em rejection surface} $s:\;\cX\to [0,1]$ that is more liberal in signal-rich regions and more stringent elsewhere, rejecting $H_i$ when $p_i \leq s(x_i)$. When the covariates are informative, covariate-assisted methods can be far more powerful than methods like the Benjamini--Hochberg (BH) procedure \citep{BH} that use a common rejection threshold for all $p_i$.

A key challenge in realizing these power gains is to avoid inflating the type I error rate despite using the data twice: first to find the signal-rich regions of $\cX$ and second to test the hypotheses. Early approaches to FDR control with informative covariates avoid FDR inflation by using fixed weights that proportionally relax the rejection threshold for {\em a priori} promising hypotheses while tightening it for others \citep{benjamini1997multiple,genovese2006false,dobriban2015optimal}, by grouping similar null hypotheses and estimating the true null proportion within each group \citep[e.g.][]{hu2010false,LFDR,liu2016new}, by ordering the hypotheses from most to least promising \citep{g2016sequential,li2017accumulation,lei2016power,cao2021optimal}, or by enforcing structural constraints on the rejection set \citep{yekutieli2008hierarchical,lynch2016procedures,STAR}. More recent approaches have sought to adaptively estimate a powerful rejection surface using all of the data, including generic covariate information as well as the $p$- or $z$-values themselves, which usually carry the best direct evidence about where the signals are \citep[e.g.][]{ferkingstad2008unsupervised,FDRreg,IHWignatiadis,BL,SABHA,bbfdr,AdaFDR}. We review and compare these methods in detail in Section~\ref{sec: Methods under comparison}.

The adaptive $p$-value thresholding (AdaPT) method of \citet{AdaPT} is a flexible and robust framework for covariate-assisted multiple testing. The user iteratively proposes a series of increasingly stringent rejection surfaces $s_1(x) \geq s_2(x) \geq \cdots$, halting the first time an FDP estimate falls below the target level $\alpha$. AdaPT offers users near-complete freedom in using the data to shape the threshold sequence, preventing FDR inflation by only allowing the analyst to observe partially masked $p$-values. Standard implementations use estimated level surfaces of the {\em local false discovery rate} (lfdr), the posterior probability that $H_i$ is true given $x_i$ and $p_i$, under an empirical Bayes two-group working model \citep{efron2008microarrays}. Because AdaPT guarantees finite-sample FDR control even when the working model is misspecified or overfit, users are free to estimate highly complex models with many predictor variables; for example, \citet{Yurko} implement AdaPT with gradient boosted trees. Indeed, AdaPT is an {\em interactive} procedure in the sense that the user may rethink their entire modeling approach based on exploratory analysis of the (masked) data midway through the procedure, without threatening the FDR guarantee.

Despite these advantages, the original AdaPT method performs poorly in two contexts: First, the finite-sample FDR control guarantee is bought at the price of a finite-sample correction that limits the method's power and stability when the rejection set is small. And second, the FDP estimator is biased upward when some of the null $p$-values concentrate near 1, as can happen when we test composite null hypotheses and some parameters lie in the interior of the null. \cite{Korthauer} remark on these shortcomings while empirically comparing the performance of various procedures including AdaPT, as do \citet{IHWignatiadis}. We discuss these issues in detail in Section~\ref{subsec: AdaPT Shortcomings}. In Section~\ref{sec: AdaPT_g} we introduce a refinement and generalization of the AdaPT procedure that we call \adaptg, which corrects AdaPT's shortcomings by modifying its masking function. 

When $z$-values are available, it is in general more efficient to use them directly rather than operating on $p$-values, especially in two-sided testing where the $p$-value transform destroys directional information \citep{sun2007oracle}. We propose a new conditional Gaussian mixture model (GMM) for the distribution of $\theta_i$ given the covariates, taking the location and scale of each component to be fixed but letting the mixing proportions vary with $x_i$. Using an expectation-maximization (EM) computation framework described in Section~\ref{sec: Implementation}, we can model this dependence on $x_i$ using any off-the-shelf classification algorithm. By modeling the test statistics directly, we can take advantage of the sign of the test statistics when the prior distribution is asymmetric, account for varying standard errors across tests, and allow for greater flexibility in testing null hypotheses other than the point null.

To demonstrate the improved reliability of the \adaptg procedure, we reproduce the empirical studies of \citet{Korthauer}, including two new methods: the AdaPT-GLM$_g$ procedure, which uses the same working model and estimation procedure as the AdaPT-GLM procedure of \citet{AdaPT} but with our new masking function, and the \gmm procedure, which replaces the GLM model for the $p$-values with our Gaussian mixture model for the $z$-statistics. Figure~\ref{fig: Case Study experiments} reproduces the main figure of \citet{Korthauer} with these two new methods added, showing substantial power gains over AdaPT in the scenarios where AdaPT fails, as well as consistent power gains over other state-of-the-art methods. We discuss these experiments in greater detail in Section~\ref{sec: Experiments}.

\subsection{Adaptive p-value Thresholding (AdaPT)}\label{subsec: AdaPT}

AdaPT is a flexible iterative framework for covariate-assisted multiple testing where the analyst proposes a series of increasingly stringent candidate rejection thresholds $s_0(x) \geq s_1(x) \geq \cdots \geq s_n(x)$, calculating an estimate of the FDP for each threshold. If the estimate at step $t$ falls below the target level $\alpha$, the method terminates and rejects every $H_i$ with $p_i \leq s_t(x_i)$; otherwise, the analyst is prompted to specify another threshold $s_{t+1}(x) \leq s_t(x)$. Figure~\ref{fig:AdaPT Masking Progression} illustrates a possible sequence of candidate rejection thresholds.

The false discovery proportion in the candidate rejection region, colored red in Figure~\ref{fig:AdaPT Masking Progression}, is estimated by comparing the number of rejections to the number of points in a mirrored region, colored blue:
\begin{equation}\label{eq:adapt-FDPhat}
\hFDP_t = \frac{1+A_t}{R_t}, \quad \text{where } R_t = |\{i:p_i\le s_t(x_i)\}| \quad \text{and } A_t = |\{i:p_i\ge 1-s_t(x_i)\}|
\end{equation}
represent the number of red and blue points, respectively. If the null $p$-values are uniform, then $A_t$ serves as a slightly conservative estimator for the number of false discoveries $V_t =|\{i: H_i \text{ is true and } p_i\le s_t(x_i)\}|$, and the addition of 1 in the numerator represents a technical finite-sample correction. If $\hFDP_t$ never falls below $\alpha$ while $R_t > 0$, we make no rejections.

\begin{figure}[t]
\centering
  \includegraphics[width=\columnwidth]{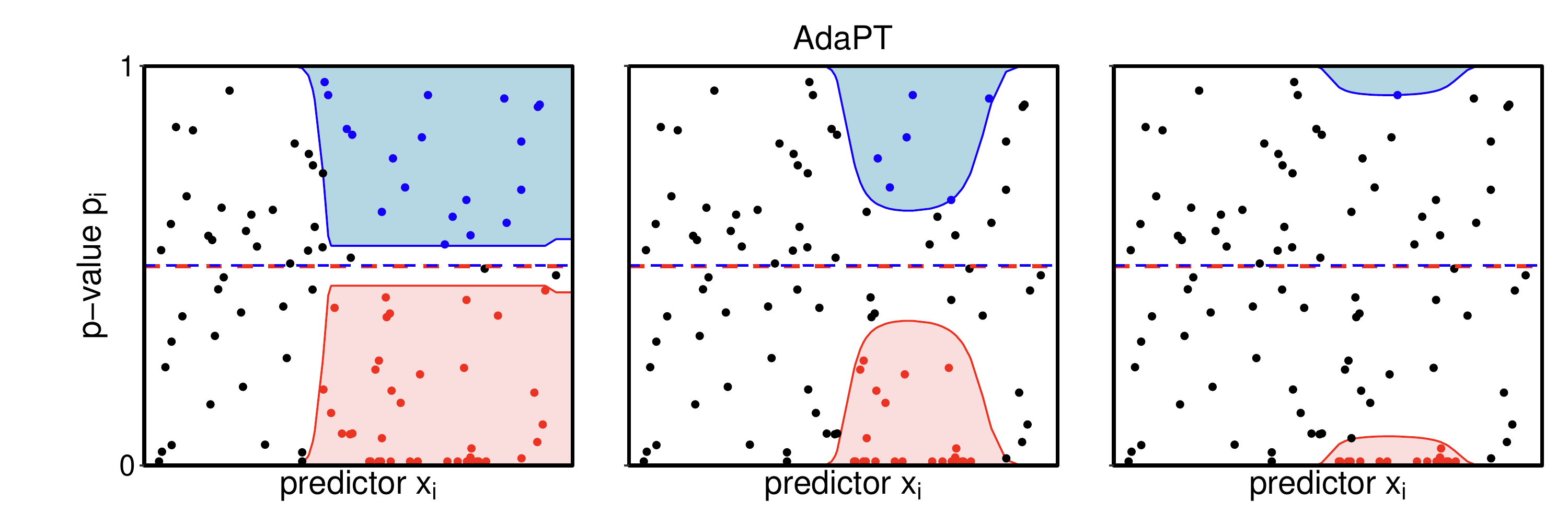}
\caption{Example progression of the AdaPT candidate rejection threshold $s_t(x)$ (red curve) for three values of $t$. $R_t$, the number of points in the red region, is the number of candidate rejections. $A_t$, the number of points in the blue mirrored region with $p_i \geq 1- s_t(x_i)$, serves to estimate $V_t$, the number of false discoveries in the red region.}
\label{fig:AdaPT Masking Progression}
\end{figure}

At step $t$ the analyst may use any data-adaptive method to choose the next threshold $s_{t+1}(x)$. In most implementations of AdaPT the thresholds are chosen to be level surfaces of the lfdr:
\begin{align}
  \text{lfdr}(p \mid x) =\PP\left(H_i\text{ is null}\mid p_i=p,x_i=x\right) \label{eq: local fdr},
\end{align}
where the conditional probabilities are calculated with respect to an empirical Bayes working model for the $p$-values called the {\em conditional two-groups model}; see \citet{AdaPT} for more details.

Because AdaPT uses the same data twice, first to select the threshold sequence and again to make rejections, it must protect against the risk of FDR inflation, which it does by strategically controlling what the analyst is allowed to observe at each step. Specifically, the analyst is initially only allowed to observe a {\em masked} version $m_i = \min\{p_i, 1-p_i\}$ of each $p$-value $p_i$. For example, if $m_i = 0.01$, then the analyst knows only that $p_i \in \{0.01, 0.99\}$. As the procedure unfolds, a $p$-value is ``unmasked'' ($p_i$ is observed) once $s_t(x_i) < m_i$ (i.e., once $p_i$ no longer contributes to $A_t$ or $R_t$). At step $t$ the analyst is allowed to observe $A_t$ and $R_t$ (to calculate $\hFDP_t$), all of the covariates $x_i$ and masked $p$-values $m_i$, and the unmasked $p$-values $p_i$ for indices in $\cM_t^\setcomp$, where $\cM_t = \{i:\; m_i \leq s_t(x_i)\}$; we call $\cM_t$ the {\em masking set}. Because $s_t$ decreases at every step, more $p$-values are unmasked as the procedure unfolds. By the end of the procedure only $A_t + R_t < (1+\alpha)R_t$ $p$-values remain masked, so that later lfdr estimates are typically calculated using almost all of the data.  We require without loss of generality that $\cM_t$ shrinks by at least one index in each step, so the procedure terminates after no more than $n$ steps.

This masking scheme is enough to prove a robust FDR control guarantee that holds regardless of how the analyst chooses to select the next threshold $s_{t+1}(x_i)$ at each step. \citet{AdaPT} show that AdaPT controls FDR at the target level $\alpha$ in finite samples, under two assumptions: First, the null $p$-values $(p_i)_{\cH_0}$ must be mutually independent of each other and of the non-null $p$-values $(p_i)_{\cH_0^\setcomp}$. $p$-values are rarely independent in practice, but there has been recent progress toward relaxing it in models where the dependence is known; see \citet{DependenceFDRFithian}. Second, for each $i\in \cH_0=\{i: H_i \text{ is true}\}$, the {\em mirror-conservative} condition must hold:
\begin{align}\label{eq: mirror conservative}
     \PP(p_i\in[a,b])\;\le\; \PP(p_i\in[1-b,1-a])\quad \text{ for all }0\le a\le b \le 0.5.
\end{align}
A sufficient condition for \eqref{eq: mirror conservative} is that $p_i$ has a non-decreasing density under $H_i$, as is the case for $p$-values from $z$- or $t$-tests, or any other continuous $p$-values from monotone likelihood ratio families \citep{STAR}.

Crucially, if the analyst uses an empirical Bayes working model to estimate an optimal threshold sequence, it is not required that the model is correctly specified or that the lfdr is estimated accurately. As a result, the analyst is liberated to use any combination of intuition, Bayesian priors, statistical estimation, or black-box machine learning to select the thresholds.

\subsection{Shortcomings of AdaPT}\label{subsec: AdaPT Shortcomings}
The AdaPT procedure is underpowered in two main situations:
\begin{enumerate}[(i)]
  \item Due to the functional form of the estimator $\widehat {\text{FDP}}_t = (1+A_t)/R_t \geq 1/R_t$, it is impossible for AdaPT to ever make fewer than $1/\alpha$ rejections (unless it makes no rejections). As a result, if only a few hypotheses are discernibly non-null, AdaPT may not be able to reject them, even if their $p$-values are exactly 0. Consequently, AdaPT is underpowered when $n$ is small, or when there are very few non-null hypotheses to find. \label{shortcomings: min hypo}
  \item If some null $p$-values concentrate at $1$, they will tend to inflate both $A_t$ and $\hFDP_t$, resulting in lower power. This may occur when we test composite null hypotheses, such as one-sided or interval null hypotheses, if some of the null parameters are located in the interior of the null parameter space. \label{shortcomings: p-values near 1}
\end{enumerate}

The first shortcoming was indirectly observed in \citet{Korthauer}, who remarked on AdaPT's low power in their simulation settings where other procedures made a small number of rejections, while \citet{IHWignatiadis} discusses both shortcomings. Observing many null $p$-values close to 1 can often confound multiple testing methods that are designed envisioning uniform null $p$-values, but properly designed methods can often improve the power relative to scenarios where the null $p$-values are uniform; see e.g. \citet{RomanoWolf, ZhaoSmallSu, ellis2020gaining,tian2019addis}.

There is an additional philosophical or practical objection one can make to the AdaPT procedure, if we are concerned about giving the researcher too many degrees of freedom:

\begin{enumerate}[(i)]
\setcounter{enumi}{2} 
  \item AdaPT can reject $p$-values greater than the nominal FDR level. In particular, it is possible for a motivated investigator to reject their favorite hypothesis $H_{i^*}$ with probability approaching 50\%, if they force $s_t(x_{i^*})$ to remain at 0.5 throughout the procedure; in that case, $H_{i^*}$ will be rejected whenever $p_{i^*} < 0.5$ and the rejection set is non-empty.  \label{shortcomings: reject > nominal}
\end{enumerate}

In some cases where $x_i$ is a highly informative predictor, standard implementations of AdaPT may estimate a very low lfdr even for $p$-values on the order of $0.15$ or $0.2$ and reject them. Whether we regard this as a feature to preserve or a bug to eliminate will depend on scientific considerations including our credence in the empirical Bayes working model.
 
Finally, we may wish to model the test statistics instead of $p$-values: 
 \begin{enumerate}[(i)]
\setcounter{enumi}{3}
 \item Whereas earlier implementations of AdaPT model the $p$-values, typically using a gamma generalized linear model (GLM) for $-\log p_i$, we may prefer modeling test statistics instead of $p$-values, especially in two-sided problems where the alternative distribution may be asymmetric. By shifting the focus of modeling to the test statistics themselves, we can also take direct account of standard errors or sample sizes that vary across hypotheses. \label{shortcomings: model p-values}
\end{enumerate}

We address the first three points in Section~\ref{subsec: Masking Function Families} by generalizing the masking function, and in Section~\ref{subsec: conditional gaussian mixture model} we introduce a new working model for directly modeling test statistics.

\section{The \adaptg Procedure}\label{sec: AdaPT_g}

\subsection{Generalizing the Masking Function}\label{subsec: Masking Function Families}
\begin{figure}[t]
\centering
\begin{subfigure}{.49\textwidth}
  \centering
  \includegraphics[width=0.9\columnwidth]{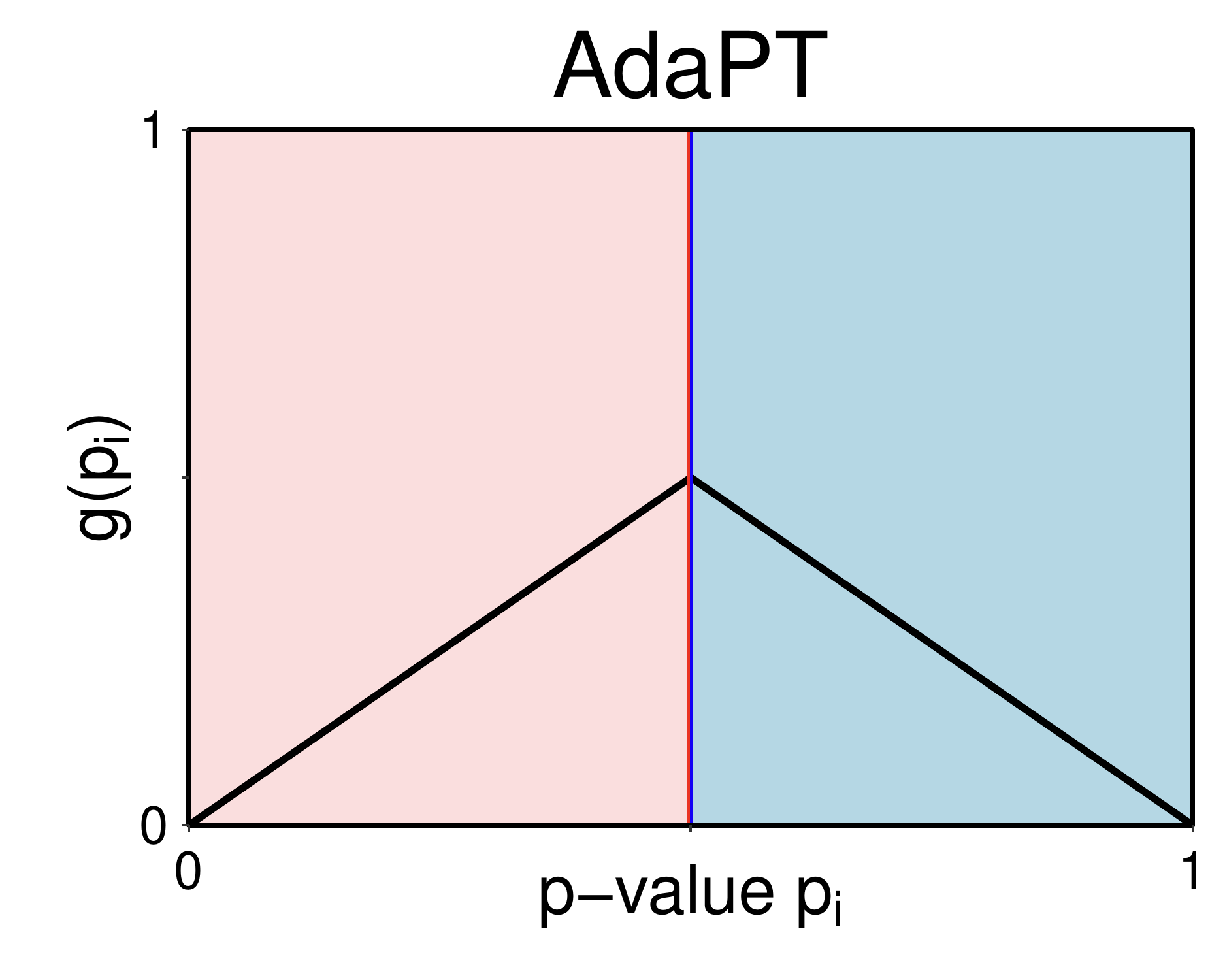}
\end{subfigure}\hfill
\begin{subfigure}{.49\textwidth}
  \centering
  \includegraphics[width=0.9\columnwidth]{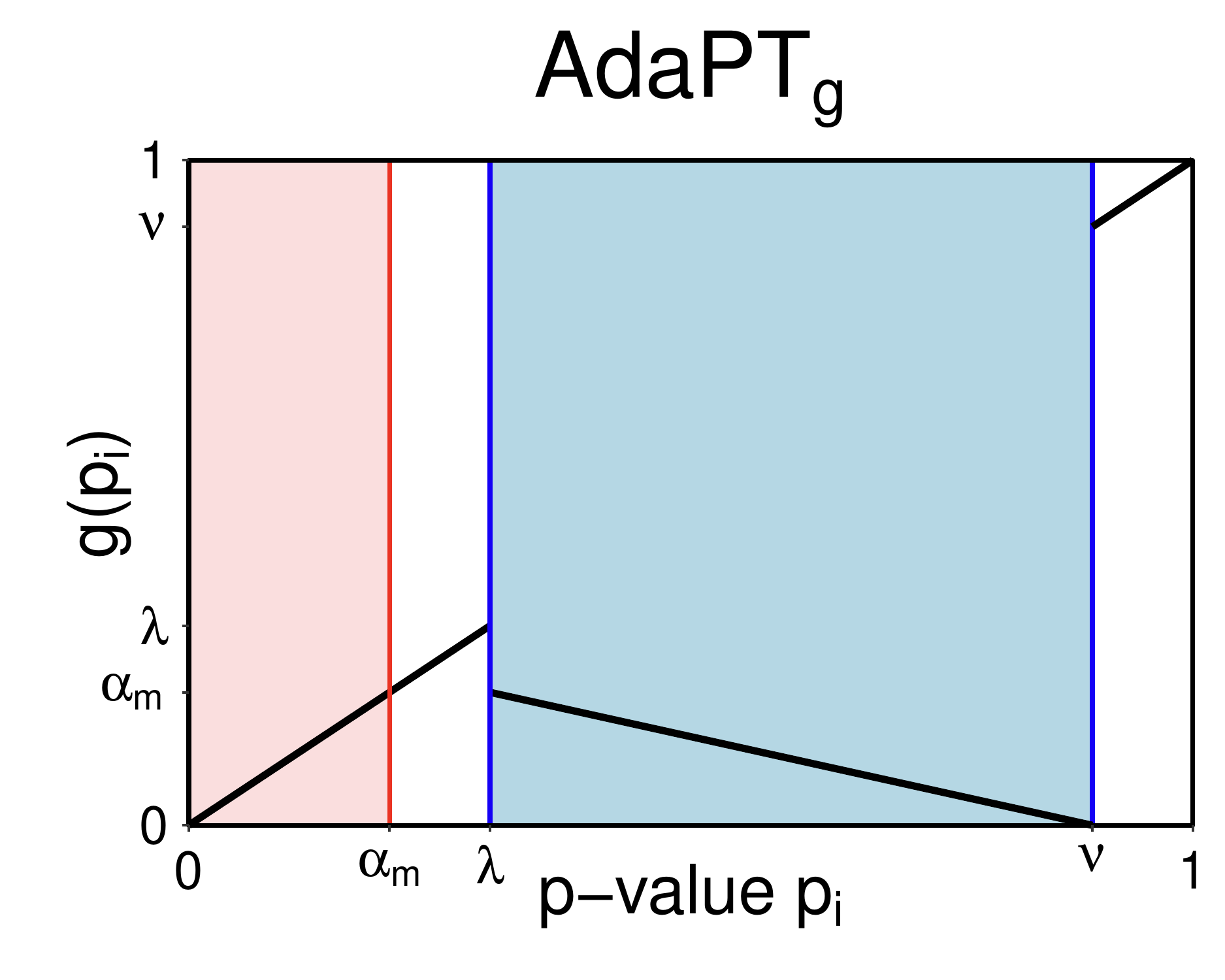}
\end{subfigure}
\caption{Masking function examples. The colored regions correspond to $p$-values that are masked and the y-axis are the masked $p$-values, what the analyst would observe. Left: Symmetric masking function for AdaPT with $\alphamax=\lambdamin = 0.5$, $\lambdamax=1$. Right: Generalized masking function with $\alphamax = 0.2$, $\lambdamin = 0.3$, $\lambdamax = 0.9$, and stretch factor $\zeta = 3$.}
\label{fig:Masking Functions}
\end{figure}

The masked $p$-value in AdaPT is the output of a two-to-one function $g(p_i) = \min\{p_i, 1-p_i\}$ whose form determines various properties of the procedure including the estimator $\hFDP_t$. Subsequent works have considered generalizations of this masking function for interactive procedures with side constraints on the rejection set \citep{STAR} and interactive FWER control procedures \citep{duan2020familywise}. The ``gap'' and ``railway'' shapes considered by \citet{duan2020familywise} are precursors to our masking function, with a similar motivation of avoiding mapping high-density regions to the same value as $g(0)$.

By generalizing the masking function, we can address the first three shortcomings of the AdaPT procedure discussed previously. We propose a new family of masking functions parametrized by three parameters, $0 < \alphamax \leq \lambdamin < \lambdamax \leq 1$, which maps a (``blue'') mirror region $[\lambdamin, \lambdamax]$ onto the initial (``red'') rejection region $[0,\alphamax]$, as pictured in Figure~\ref{fig:Masking Functions}. Defining the {\em stretch parameter} $\zeta=(\lambdamax-\lambdamin)/\alphamax$, the ratio between the widths of the two regions, the masking function is defined as
\begin{equation}\label{eq: General Masking Function g(p) definition}
    g(p) = \begin{cases} 
    (\lambdamax - p)/\zeta & p \in [\lambdamin , \lambdamax]\\
    p & \text{otherwise}
    \end{cases}.
\end{equation}
If we take $\alphamax = \lambdamin = 0.5$ and $\lambdamax = 1$, then $\zeta = 1$ and we recover the initial symmetric AdaPT procedure. 

Equation~\eqref{eq: General Masking Function g(p) definition} defines $g(p)$ so that $g(0)=g(\nu)$, creating the ``tent'' shape pictured in Figure~\ref{fig:Masking Functions}. Alternatively we could replace $(\lambdamax - p)/\zeta$ with $(p - \lambdamin)/\zeta$ in the definition of $g$, resulting in  a ``comb'' shape such that $g(0)=g(\lambda)$ instead. We prefer the tent shape for most applications, but use the comb shape for testing interval nulls; we discuss this choice in more detail in Appendix~\ref{app: masking function shapes}.

Figure~\ref{fig:Masking Functions} illustrates the symmetric AdaPT masking function and our generalized masking function side by side. Note that for $p \in (\alphamax, \lambdamin) \cup (\lambdamax, 1]$, $g(p)$ is one-to-one, so $p$-values in that region are never masked from the analyst. For $p_i\in [0,\alphamax] \cup [\lambdamin, \lambdamax]$, the masking function effectively ``splits'' each $p$-value $p_i$ into the masked value $m_i = g(p_i)$ and the binary variable $b_i =  \mathbbm{1}\{p_i \in [\lambdamin, \lambdamax]\}$, which indicates whether $p_i$ is the larger of the two values that $g$ maps to $m_i$. If $p_i$ is uniformly distributed, then $b_i \mid x_i, m_i \;\sim\; \text{Bern}(\zeta/(1+\zeta))$.

Let $p_{i,b}(m_i)$ denote the implied value of $p_i$ if $b_i = b \in \{0,1\}$. For $m_i \leq \alphamax$, $p_{i,0}=m_i$ and $p_{i,1} = \nu - \zeta m_i$; for example, in the masking function pictured in Figure~\ref{fig:Masking Functions}, if $m_i = 0.01$ then $p_i$ is either $p_{i,0} = 0.01$ or $p_{i,1} = 0.9 - 3 \cdot 0.01 = 0.87$. For $m_i > \alphamax$, $b_i=0$ almost surely and $p_{i,1}$ is undefined. 

In the $z$-test setting with $z_i \sim \cN(\theta_i, \sigma_i^2)$, we will want an analogous definition of $z_{i,b}$. Let $\Phi(z)$ denote the standard Gaussian cumulative density function. To test the one-sided hypothesis $H_i:\;\theta_i \leq 0$, we have $p_i = 1 - \Phi(z_i/\sigma_i)$, a one-to-one transform, so we can straightforwardly define $z_{i,b} = \sigma_i\Phi^{-1}(1-p_{i,b})$. For two-sided testing of the point-null $H_i:\; \theta_i = 0$ or the interval null $H_i:\; |\theta_i| \leq \delta$, however, there are four possible $z$-values mapping to the same $m_i$, since each of $p_{i,0}$ and $p_{i,1}$ could result from either a positive or a negative $z$-value. For the point null, we can resolve this issue by also revealing $s_i = \sgn(z_i)(-1)^{b_i}$ at the very beginning of the procedure, which narrows the four options down to two and restores the one-to-one relationship between $p_i$ and $z_i$. If $\theta_i = 0$ then $s_i$ is conditionally independent of $b_i$ given $m_i$, so revealing this information to the analyst does not affect the FDR control guarantee. For the composite interval null, however, we cannot rely on any such conditional independence property, so we must retain all four options in our model fitting. We discuss these implementation details further in Appendix~\ref{app: interval null testing}.

\subsection{The \adaptg Procedure}\label{subsec: adapt_g implementation}
\begin{figure}[t]
\centering
  \includegraphics[width=\columnwidth]{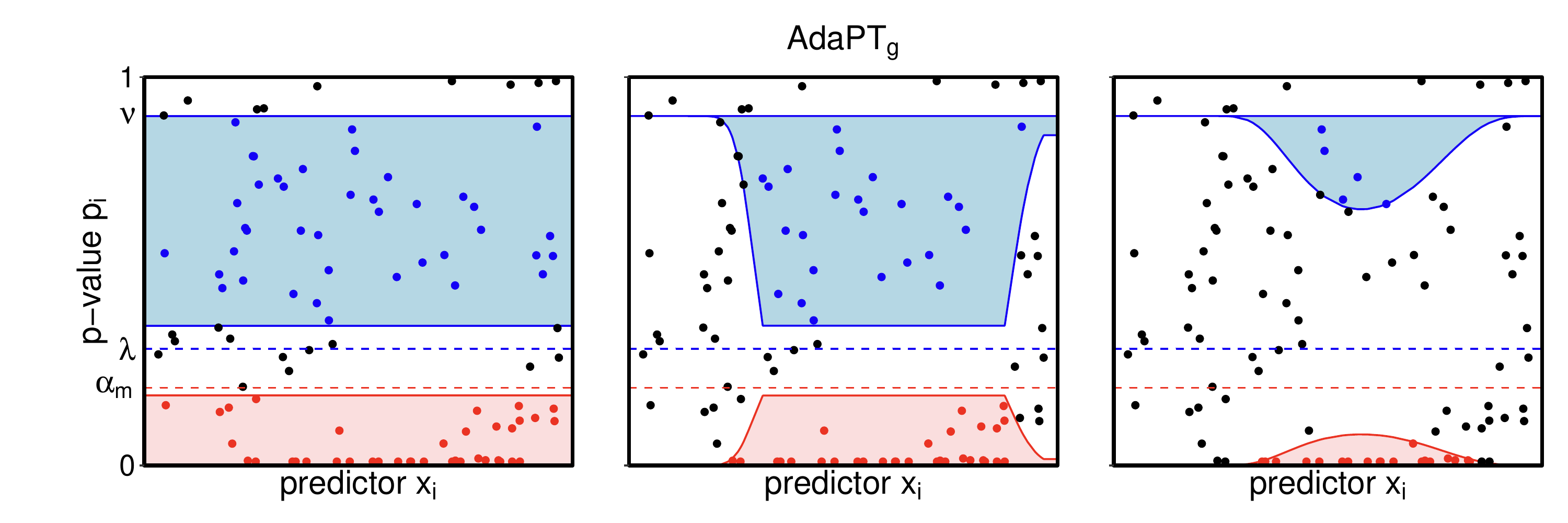}
\caption{Example progression of the \adaptg candidate rejection threshold $s_t(x)$, compared to Figure~\ref{fig:AdaPT Masking Progression}. The blue mirror region is the red region reflected vertically and stretched vertically by a factor $\zeta$. $A_t/\zeta$, the number of blue points divided by the stretch factor, serves to estimate $V_t$, the number of false discoveries among the candidate rejections (red points).}

\label{fig:gmm Masking Progression}
\end{figure}

Like the original AdaPT procedure, the more general \adaptg procedure also estimates FDP for a sequence of increasingly stringent thresholds $s_t(x),$ beginning with the constant threshold $s_0(x) \equiv \alphamax$:
\begin{equation} \label{eq: generalized FDP hat, A_t, and R_t}
\hFDP_t = \frac{1+A_t}{\zeta R_t}, \quad \text{where } R_t = |\{i:\,p_i\le s_t(x_i)\}| \quad \text{and } A_t = |\{i:\, m_i \le s_t(x_i), \,b_i = 1\}|. 
\end{equation}

Figure~\ref{fig:gmm Masking Progression} shows an example progression of the \adaptg procedure analogous to Figure~\ref{fig:AdaPT Masking Progression}, but with the generalized masking function. As before, $R_t$ is the number of ``red'' candidate rejections and $A_t$ is the number of points in the ``blue'' mirror region, now stretched by a factor $\zeta$. More explicitly, $i$ contributes to $A_t$ if $\lambdamax - \zeta s_t(x_i)\le p_i \le \lambdamax$. The factor of $\zeta$ in the denominator reflects the stretching; heuristically, we now have $V_t \approx A_t/\zeta$ since a uniform $p$-value is $\zeta$ times as likely to be ``blue'' than ``red.'' Formally, the estimator is based on the general version of selective SeqStep defined in \citet{barber2015controlling}. The minimum nonzero number of rejections at level $\alpha$ is now $R_{\min} = (\zeta\alpha)^{-1}$, since we can halt when $A_t = 0$ and $R_t \geq R_{\min}$.

As in the original AdaPT procedure, at step $t$ the analyst is only allowed to observe $A_t$ and $R_t$, all $x_i$ and $m_i$ values, and $p_i$ for $i \in \cM_t^\setcomp$, where as before $\cM_t = \{i:\; m_i \leq s_t(x_i)\}$, so that only observations contributing to $A_t$ or $R_t$ are masked. Equivalently, we can say the analyst observes $A_t$ and $R_t$, all $x_i$ and $m_i$ values, and $b_i$ for $i \in \cM_t^\setcomp$. 

To see why the new masking scheme can resolve AdaPT's first three shortcomings, suppose we take $\alphamax = 0.05$, $\lambdamin = 0.4$, and $\lambdamax = 0.9$, giving a stretch factor $\zeta = 10$. Then, if we are controlling FDR at level $\alpha = 0.1$,
\begin{enumerate}[(i)]
\item the small-sample issue is resolved because $R_{\min} = (10 \cdot 0.1)^{-1} = 1$, so we are able to make any number of rejections;
\item null $p$-values in $(0.9, 1]$ do not contribute to $A_t$, so a null $p$-value density spike at 1 does not inflate the FDP estimate; and
\item all rejection thresholds are uniformly no higher than $0.05$, so no individual null $H_i$ can be rejected with probability higher than $0.05$.
\end{enumerate}

In our view, these are suitable default parameter choices for a conservative scientist who wishes to insist on strong individual evidence against each rejected hypothesis. A user interested primarily in maximizing the power may prefer to increase $\alphamax$, in which case the tradeoffs between the parameters must be weighed more carefully. 

On one hand, choosing a large $\zeta$ improves the FDP estimation in addition to reducing $R_{\min}$. To see why, suppose that $x_i \simiid P_x$, the conditional null proportion is $\pi_0(x) = \PP(H_i \text{ is true } \mid x_i=x)$, and the null $p$-values are i.i.d. uniform, and consider running \adaptg with a fixed (non-adaptive) threshold sequence. Then the expected number of false discoveries at step $t$ is 

\[
\lambda_t \;=\; \EE V_t \;=\; n \int_\cX \pi_0(x) s_t(x)\,dP_x(x),
\]
and $V_t \sim \text{Binom}(n, \lambda_t/n) \approx \text{Pois}(\lambda_t)$ if $\lambda_t \ll n$, as we typically expect. If null $p$-values predominate in the blue region, we likewise have $A_t \approx \text{Pois}(\zeta \lambda_t)$, so that 
\[
\EE\left[\frac{A_t+1}{\zeta} \right] \;\approx\; \lambda_t + \zeta^{-1}, \quad \text{ and } \; \text{Var}\left(\frac{A_t+1}{\zeta} \right) \;\approx\; \lambda_t/\zeta.
\]

Thus, increasing $\zeta$ tends to reduce both the bias and the variance of the FDP estimator.
On the other hand, larger values of $\lambdamax$ risk inflating $A_t$ when some $p$-values are super-uniform, and smaller values of $\lambdamin$ both limit how large we can take $\alphamax$, and risk including more alternative hypotheses in the blue region. If we expect many rejections and want to hold open the possibility of rejecting even relatively large $p$-values, we would prioritize increasing $\alphamax$ at the price of reducing $\zeta$, whereas if we expect to reject only a few hypotheses with very small $p$-values, we will tend to prioritize increasing $\zeta$. We provide more aggressive default recommendations in Section~\ref{subsec: Adaptive Masking Functions} for scientists who wish to gain power by increasing $\alphamax$.

As a naming convention in this paper, we denote generalized versions of the AdaPT procedure with this new masking scheme using a $g$ subscript; for example, we call AdaPT-GLM with the new masking function $\text{AdaPT-GLM}_g$ to distinguish it from the previous version that used $g(p)=\min\{p,1-p\}$. As we will see, the new masking consistently improves AdaPT's reliability and power, so asymmetric masking functions will henceforth be the default for versions 2.0 and higher of the \texttt{adaptMT} package.

\subsection{\adaptg without Thresholds}

As observed by \citet{AdaPT} in their discussion of ``AdaPT without thresholds,'' the only role that the threshold sequence plays in the method is in determining which $p$-values are masked at each step, and therefore which $p$-values contribute to $A_t$ and $R_t$. Instead of prompting the analyst for a threshold sequence, we could simply let the analyst choose adaptively at step $t$ which $p$-values to unmask for step $t+1$, giving
\[
[n]\supseteq \cM_0\supsetneq \cM_1\supsetneq \cdots\supsetneq  \cM_n = \emptyset.
\]

More generally, we can implement \adaptg without thresholds for any sequence of masking sets if we augment $x_i$ with the index $i$ and use the threshold sequence $s_t(x,i) = \alphamax \cdot 1\{i \in \cM_t\}$. As long as $\cM_{t+1} \subsetneq \cM_t$ is selected using the available information at step $t$, any such sequence is a valid implementation of \adaptg. Conversely any threshold sequence results in a nested sequence of masking sets, so the two formulations are equivalent. Algorithm~\ref{algo: AdaPT no thresh} formally defines the method, generalizing Algorithm 3 in \citet{AdaPT}.

\begin{algorithm}
  \caption{AdaPT$_g$ without Thresholds}
  \label{algo: AdaPT no thresh}
  \textbf{Input:} predictors and $p$-values $(x_i,p_i)_{i\in[n]}$, masking parameters $\alphamax,\lambdamin,\lambdamax$, target FDR level $\alpha$.\\
  Initialize: $\mathcal{M}_0=\{i\in[n]:p_i\in[0,\alphamax] \text{ or } p_i\in[\lambdamin,\lambdamax]\}$ and $\{m_i\}_{i \in [n]}= \{g(p_i)\}_{i \in [n]}$
  \begin{algorithmic}[1]
    \For {$t=0,1,\ldots$}
      \State{$\widehat{\text{FDP}}_t\leftarrow \frac{1+A_t}{\zeta R_t}$}
      \If{$\widehat{\text{FDP}}_t\le \alpha$}
      \State{Reject $R_t$}
      \EndIf
      \State{$\mathcal{M}_{t+1}\leftarrow {\scriptstyle\mathrm{UPDATE}}((x_i,m_i)_{i\in[n]},\mathcal{M}_{t},\{p_i: i\in \cM_t^\setcomp\})$}

      \State{$R_{t+1}\leftarrow | \{H_i:i \in \mathcal{M}_{t+1} \text{ and } 0 \le p_i\le \alphamax \}| $}
      \State{$A_{t+1}\leftarrow| \{H_i:i \in \mathcal{M}_{t+1} \text{ and } \lambdamin \le p_i\le \lambdamax \}| $}
    \EndFor
  \end{algorithmic} 
\end{algorithm}

\subsection{FDR Control and Power for \adaptg} \label{subsec: FDR Control and revealing hypotheses}

As we see next, \adaptg has a similarly robust FDR guarantee as the original AdaPT procedure.

\begin{theorem}\label{theorem: fdr control}
Assume that the null $p$-values are mutually independent and independent of the non-null $p$-values, and assume that the null $p$-values have non-decreasing density. Then the \adaptg procedure controls the FDR at the target level $\alpha$. 
\end{theorem}

The proof, which generalizes the FDR control proof for the original AdaPT procedure, is in Appendix~\ref{app: FDR control}. Note the condition on the null $p$-value distributions has been strengthened to require a non-decreasing density under the null. This condition is satisfied by continuous one- or two-sided $p$-values in monotone likelihood ratio families, including the Gaussian distribution. 

As with AdaPT, note that the FDR control guarantee for \adaptg holds in finite samples and does not require that we estimate ldfr consistently, or even that we estimate lfdr at all: for any decreasing threshold sequence, or equivalently any nested sequence of masking sets, FDR is controlled.

To maximize power, at each step the analyst should try their best to reveal ``blue'' rather than ``red'' observations. A natural strategy is to employ some empirical Bayes working model to estimate
\begin{align}
  q_i \;=\; \PP[b_i=1\mid x_i,m_i] \;=\; \frac{\zeta f(p_{i,1} \mid x_i)}{f(p_{i,0} \mid x_i) + \zeta f(p_{i,1} \mid x_i)}, \label{eq: probability of big}
\end{align}
where $f(p \mid x)$ is the mixture distribution of the $p$-values given $x$, and then reveal the $p$-value with the largest estimate $\hq_{i,t}$ given the information available at time $t$:
\begin{equation}\label{eq:strategy}
\cM_{t+1}=\cM_{t}\setminus\{\hat{i}_t\}, \quad \text{ where } \; \hat{i}_t= \argmax_{i\in \cM_t} \hq_{i,t}.
\end{equation}
In case of a tie, we can choose $\hat{i}_t$ arbitrarily from the $\argmax$.

Theorem~\ref{theorem: optimal revealing with oracle} shows that, if we could calculate these conditional probabilities without estimation error, this rule gives the most powerful sequence of masking sets among all adaptive strategies.

\begin{theorem} \label{theorem: optimal revealing with oracle}
The oracle version of the strategy in \eqref{eq:strategy}, where the true probabilities $q_i$ replace the estimates $\hq_{i,t}$ and
\[
\cM_{t+1}=\cM_{t}\setminus\{i_t^*\}, \quad \text{ where } \; i_t^*= \argmax_{i\in \cM_t} q_i
\]
gives the most powerful sequence of masking sets. That is, it maximizes $\PP[R \geq r]$ for every $r = 1,\ldots,m$ over all possible adaptive strategies for shrinking the masking set, where $R$ is the number of rejections made by the generalized AdaPT procedure.
\end{theorem}

Theorem~\ref{theorem: optimal revealing with oracle} is proved in Appendix~\ref{pf: optimal revealing with oracle}.

\subsection{Conditional Gaussian Mixture Model}\label{subsec: conditional gaussian mixture model}

We now introduce a flexible new class of working models for the common setting where $z_i \sim \cN(\theta_i, \sigma_i^2)$, and $H_i$ concerns $\theta_i$. Instead of modeling the $p$-values, we model the $z$-values directly, which allows the model to naturally incorporate details such as variability in the standard errors and asymmetry in the alternative distribution. In addition, it is naturally adaptable to the goal of testing composite null hypotheses such as one-sided and interval nulls.

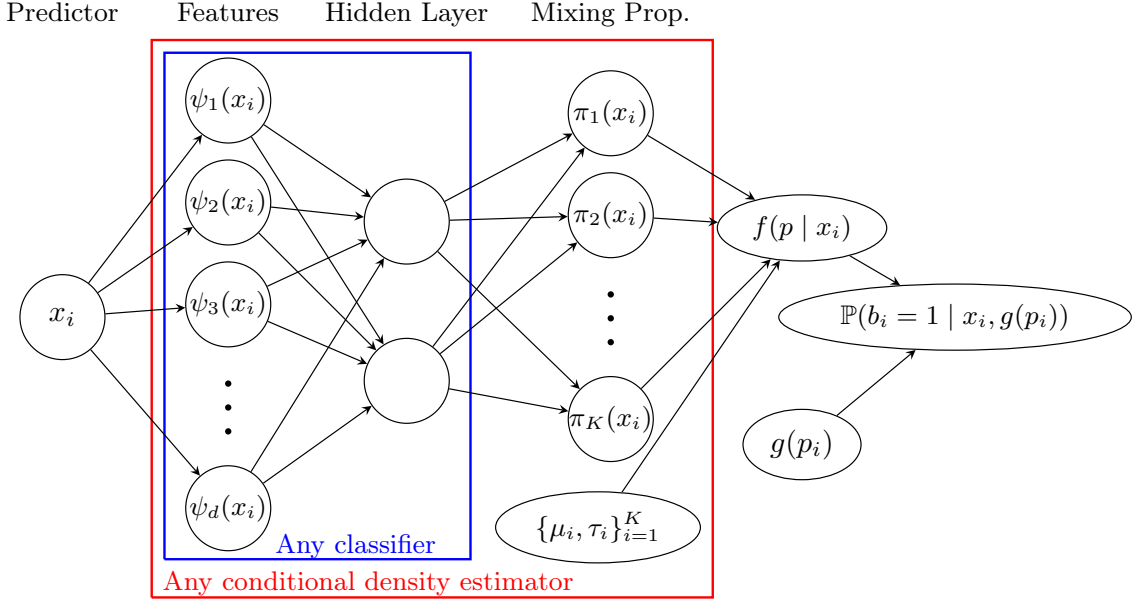
\begin{figure}[h]
    \centering
    \resizebox{\textwidth}{!}{%
    \centering
    \tikzset{%
  every neuron/.style={
    circle,
    draw,
    minimum size=1cm
  },
  every ellipse/.style={
    ellipse,
    draw,
    minimum size=0.7cm
  },
  neuron missing/.style={
    draw=none, 
    scale=2,
    text height=0.333cm,
    execute at begin node=\color{black}$\vdots$
  },
}

\begin{tikzpicture}[x=1.5cm, y=1.5cm, >=stealth]
\node [every neuron](x) at (0,0) {$x_i$};

\foreach \m [count=\y] in {1,2,3,missing,4}
  \node [every neuron/.try, neuron \m/.try] (input-\m) at (1.3,2.5-0.8*\y) {};

  \foreach \m [count=\y] in {1,2,3,d}
  \node[] at(input-\y) {\small ${\psi_\m(x_i)}$};

\foreach \m [count=\y] in {1,2}
  \node [every neuron/.try, neuron \m/.try ] (hidden-\m) at (2.7,2-\y*1.25) {};

\foreach \m [count=\y] in {1,2,missing,3}
  \node [every neuron/.try, neuron \m/.try] (mixing-\m) at (4.3,2.4-0.8*\y) {};

  \foreach \m [count=\y] in {1,2,K}
  \node [] at (mixing-\y) {\small$\pi_\m(x_i)$};


\node [every ellipse](f) at (5.8,0.7) {\small $f(p\mid x_i)$};
\node [every ellipse](mu_tau) at (4.2,-1.65) {\small$\{\mu_i,\tau_i\}_{i=1}^K$};

\node [every ellipse ](prob_blue) at (7,0) {\small$\PP(b_i=1\mid x_i,g(p_i))$};
\node [every ellipse ](g) at (5.8,-1) {$ g(p_i)$};

\draw[blue,thick] (0.8,-1.9) -- (3.2,-1.9) -- (3.2,2.075) -- (0.8,2.075) -- (0.8,-1.9);
\node [blue ] at (2.3,-1.8) {\small Any classifier};
\draw[red,thick] (0.7,-2.2) -- (5.1,-2.2) -- (5.1,2.175) -- (0.7,2.175) -- (0.7,-2.2);
\node [red ] at (2.4,-2.09) {\small Any conditional density estimator};




\foreach \i in {1,...,4}
    \draw [->] (x) -- (input-\i);

\foreach \i in {1,...,4}
  \foreach \j in {1,...,2}
    \draw [->] (input-\i) -- (hidden-\j);

\foreach \i in {1,...,2}
  \foreach \j in {1,...,3}
    \draw [->] (hidden-\i) -- (mixing-\j);

  \foreach \i in {1,...,3}
    \draw [->] (mixing-\i) -- (f);

  \draw [->] (mu_tau) -- (f);
\draw [->] (f) -- (prob_blue);
\draw [->] (g) -- (prob_blue);

\node []  at (0, 2.4) {\small Predictor};
\node []  at (1.3, 2.4) {\small Features};
\node []  at (2.7,2.38) {\small Hidden Layer};
\node []  at (4.3,2.38) {\small Mixing Prop.};

\end{tikzpicture}
    }
    \caption{Graphical representation of the interaction between covariates, classifer, Gaussian mixture model, and computed conditional density of $p$-values.}
    \label{fig: flowchart}
\end{figure}

We model the conditional distribution of $\theta_i$ given predictor $x_i$ as a Gaussian mixture model (GMM) with $K$ classes, where the class probabilities depend on $x_i$:
\[
f(\theta_i \mid x_i) \;\sim\; \sum_{k=1}^K \pi_k(x_i) \phi(\theta_i; \mu_k, \tau_k^2), \quad \text{ where } \;\;\phi(\theta; \mu, \tau^2) \,=\, \frac{1}{\sqrt{2\pi \tau^2}} \,e^{-(\theta-\mu)^2/(2\tau^2)}
\]
is the $\cN(\mu,\tau^2)$ density. Note that the location and scale of the mixture components do not depend on $x_i$, but the overall distribution can shift and stretch as $x_i$ changes by varying the mixing proportions.

We emphasize again that our FDR control guarantee does not rely on this model to be correctly specified or estimated accurately. Because deconvolution is a very difficult statistical problem, and we can expect only a small number hypotheses to be discernibly non-null in any given application, we cannot realistically expect our model for $f(\theta_i \mid x_i)$ to closely mirror the true data-generating distribution, even when the marginalized model for $f(z_i \mid x_i)$ fits fairly well. In particular, we will often estimate that most of the data come from a single component with $\mu_k, \tau_k^2 \approx 0$. In that case, any lfdr estimates for the one-sided null will depend sensitively on whether that $\mu_k$ is just above zero or just below zero, which is nearly impossible to estimate from the data. For this reason, our algorithm relies on our estimate for $f(z_i \mid \theta_i)$ instead of our lfdr estimates.

To facilitate estimation, we introduce the latent categorical variable $\gamma_i\in[K]$ to represent which of the $K$ classes $\theta_i$ is drawn from, leading to the hierarchical model
\begin{align}\label{eq: hierarchical model}
\begin{split}
  \PP(\gamma_i = k \mid x_i) &\;=\; \pi_k(x_i)\\
  \theta_i \mid x_i, \gamma_i = k &\;\sim\; \mathcal{N}(\mu_k, \tau_k^2)\\[1.0ex]
  z_i \mid x_i, \gamma_i, \theta_i &\;\sim\; \mathcal{N}(\theta_i, \sigma_i^2).
  \end{split}
\end{align}

We refer to the implementation of AdaPT with the asymmetric masking function and the K-groups GMM as \gmm. Figure~\ref{fig:CGMM model} shows an example of the fitted distribution for one run of the logistic mixture simulation from Section~\ref{subsec: logistic simulations}.

\begin{figure}
    \centering
    \includegraphics[width=0.45\columnwidth]{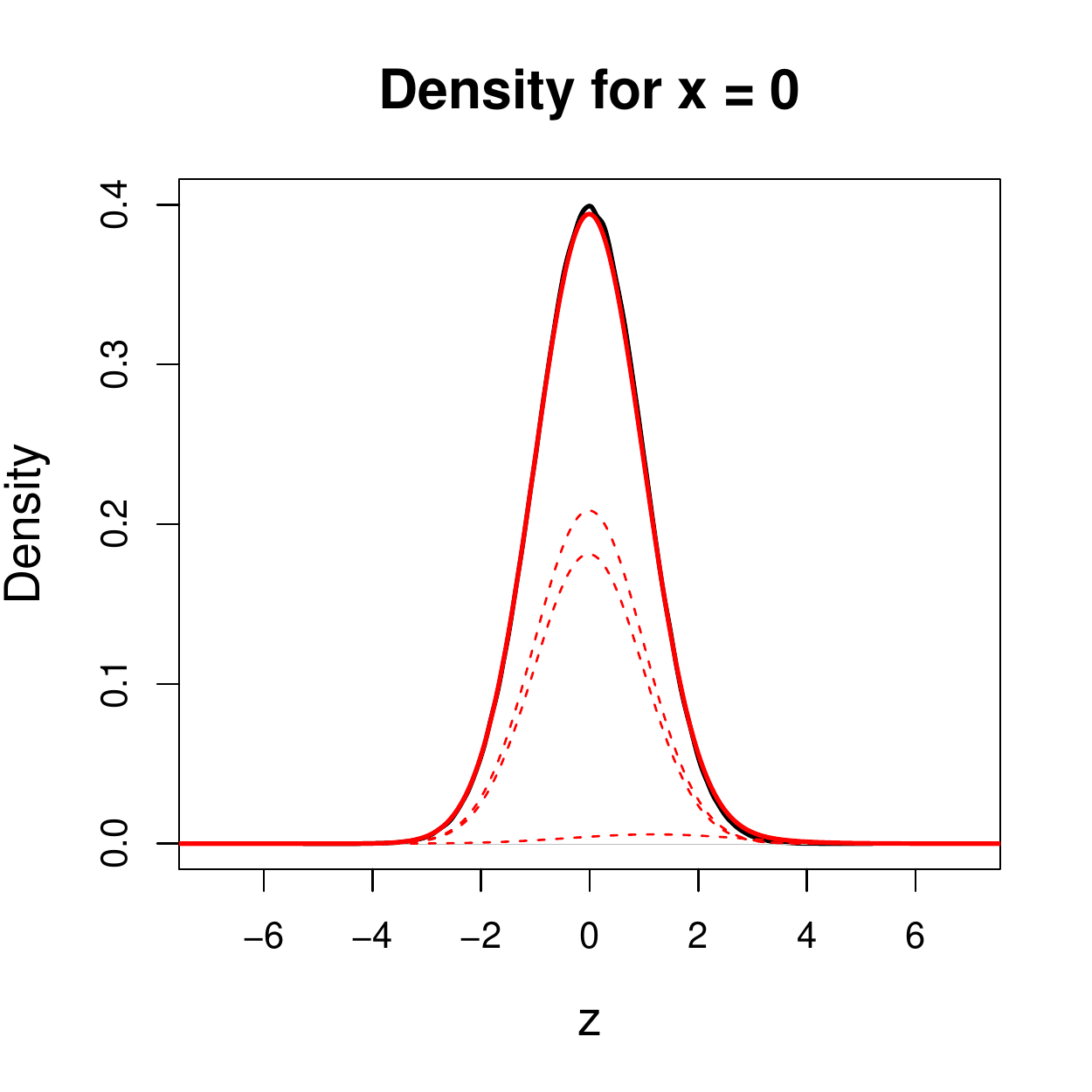}
    \includegraphics[width=0.45\columnwidth]{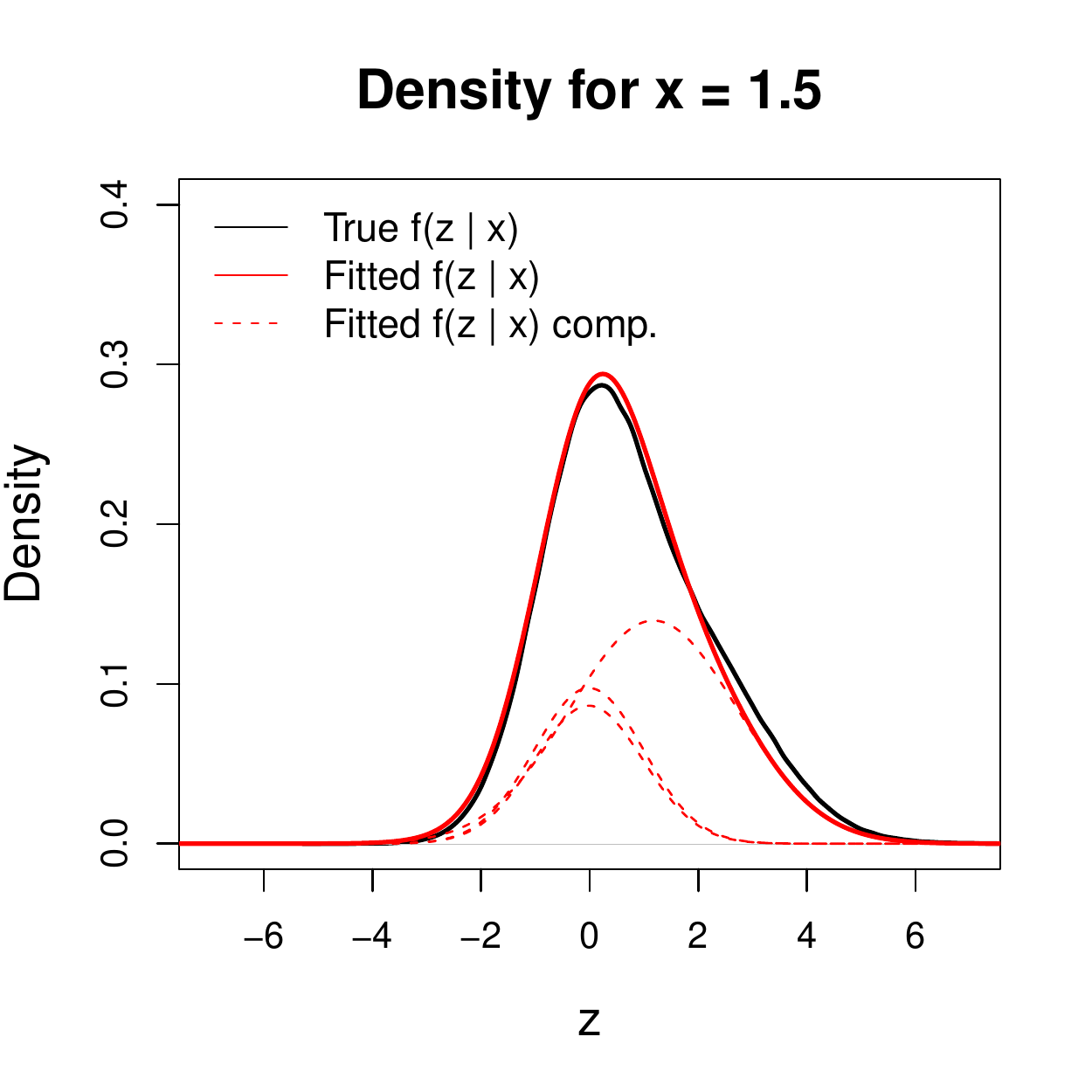}\\
    \caption{Fitted distributions for the density of $z_i$ given $x_i$ for the logistic simulations in Section~\ref{subsec: logistic simulations}. The model estimates three mixture components. Larger values of $x_i$ increases the probability of drawing from the right tail.}
    \label{fig:CGMM model}
\end{figure}

As we will see in Section \ref{sec: Implementation}, we can estimate the functional dependence of $\pi_k$ on $x_i$ by using any off-the-shelf classifier as a module in our flexible EM optimization scheme, provided the classifier accepts weighted observations and outputs class probabilities. In selecting a classifier, we should keep in mind that in multiple testing problems there is typically much less information available to learn complex dependencies on $x_i$ than the nominal ``sample size'' $n$ might suggest, since $\gamma_k$ is only observed indirectly through $\theta_i$ and $z_i$, and the vast majority of $\theta_i$ values are typically indistinguishable from 0. As a result we should usually aim to estimate models with relatively few degrees of freedom.

To this end, we use a neural network model with a single hidden layer as our default modeling option, with a user-specified featurization $\psi(x_i) \in \RR^d$. If the hidden layer has $h$ nodes, then the neural network has $(d + K - 1)h$ parameters to estimate; by contrast, a standard multinomial logit model must estimate $d(K-1)$ parameters. If, say, $K=6, h=2,$ and $d=10$, then the hidden layer does an effective job of economizing on model degrees of freedom.

Depending on the details of the problem, we may wish to further reduce our modeling degrees of freedom by forcing the distribution to be symmetric. If so, we implement the symmetry assumption by replacing each $N(\mu_k, \sigma_k^2)$ mixture component with a mixture of $N(\mu_k, \sigma_k^2)$ and $N(-\mu_k, \sigma_k^2)$, assigning $\pi_k(x_i)/2$ weight to each component. We observe modest performance gains by enforcing symmetry when the data distribution is roughly symmetric.

In some applications, we might also wish to reparameterize the test statistics as standardized $z$-values $z_i' = z_i/\sigma_i$; then the point null and one-sided null can be equivalently stated in terms of the standardized effect size $\theta_i' = \theta_i/\sigma_i$, and $z_i'$ can used as inputs to the method with unit variance. This decision mostly comes down to which of $\theta_i$ or $\theta_i'$ we expect is more likely to follow a predictable distribution given $x_i$; in either case we can use $\sigma_i^2$ as a predictor variable. Finally, if we want to use \gmm in a problem for which only $p$-values are available, we can map $p_i \mapsto z_i = \Phi^{-1}(1-p_i)$ and input the latter as right-tailed $z$-values with unit variance. In all of our empirical studies, we use the symmetric method with standardized $z$-values as inputs when $z_i$ and $\sigma_i^2$ are available, and otherwise we map the $p$-values to right-tailed $z$-values.

\section{Implementation}\label{sec: Implementation}

\subsection{Expectation-Maximization Algorithm}

To implement the conditional GMM of Section~\ref{subsec: conditional gaussian mixture model}, we use an expectation-maximization (EM) algorithm with a generic classifier module in the ``M-step.'' The observed data at step $t$ are $x_i$ and $m_i$ for each $i$ as well as $b_i$ for the currently unmasked $p$-values, $i\in \cM_t^\setcomp$; we treat $b_i$ as missing data when $i \in \cM_t$. In addition, we introduce latent variables $\gamma_i$ for each $i$, representing the mixture component responsible for $z_i$. While $A_t = \sum_{i\in \cM_t} b_i$ and $R_t = |\cM_t|-A_t$ are also observed by the analyst, we ignore them here and estimate the parameters as though $A_t$ and $R_t$ were unknown, since conditioning on the sum of the masked $b_i$ values introduces a substantial complication to the algorithm for a minimal improvement in estimation performance.

Rather than introducing the true parameters $\theta_i$ as additional latent variables, we can simply marginalize over them in \eqref{eq: hierarchical model} to obtain the reduced-form model
\begin{align}\label{eq: reduced model}
\begin{split}
  \PP(\gamma_i = k \mid x_i) &\;=\; \pi_k(x_i)\\
  z_i \mid x_i, \gamma_i = k &\;\sim\; \mathcal{N}(\mu_k, \,\tau_k^2 + \sigma_i^2)\\
  \end{split}.
\end{align}

For simplicity we will restrict our discussion to testing one-sided hypotheses $H_i:\; \theta_i \leq 0$, so that $p_{i,b} = 1 - \Phi(z_{i,b}/\sigma_i)$; we discuss point and interval null hypotheses in Appendix~\ref{app: two sided testing} and~\ref{app: interval null testing}.

Let $\beta$ denote a generic parameter vector for the class probability model, and denote the full set of parameters for the GMM as $\eta = (\beta, \mu_1, \tau_1^2, \ldots, \mu_K, \tau_K^2)$. The complete data log-likelihood, observing all $b_i$ and $\gamma_i$, is  
\begin{align*}
\ell(\eta; \;\gamma, b, m, x) 
    &= \log \left\{\prod_{i=1}^n \pi_{\gamma_i}(x_i; \beta) \phi(z_{i,b_i}; \mu_{\gamma_i}, \tau_{\gamma_i}^2 + \sigma_i^2)\right\}\\[1.0ex]
    &= \sum_{i=1}^n \sum_{k=1}^K \sum_{b=0}^1 
    \mathbbm{1}\{\gamma_i=k,b_i=b\}\left\{\log \pi_k(x_i; \beta) + \log \phi(z_{i,b}; \mu_{k}, \tau_{k}^2 + \sigma_i^2)\right\}.
  \label{eq: complete loglikelihood} \numberthis
\end{align*}
At EM iteration $s$ for step $t$ of AdaPT we choose the next estimate $\heta^{(s+1,t)}$, the parameters for the next iteration, to maximize the conditional expectation of $\ell(\eta)$, under the current parameter estimate $\heta^{(s,t)}$. To that end, in the ``E step'' we will calculate the conditional probabilities
\[
w_{ikb}^{(s,t)} = \mathbb{P}_{\hat\eta^{(s,t)}}\left(\gamma_i = k, b_i = b \mid (x_j, m_j)_{j\in[n]}, (b_j)_{j\in\mathcal{M}_t^\setcomp}  \right).
\]
Note some of these weights are zero, for example if the analyst has observed $b_i=0$ by step $t$ then $w_{ik1}^{(s,t)} =0$. Furthermore let $w_{ik+}^{(s,t)} = \sum_{b=0}^1 w_{ikb}^{(s,t)}$. After taking conditional expectations, the ``M step'' of \eqref{eq: complete loglikelihood} reduces to
\begin{align}\label{eq: EM max problem}
\begin{split}
\hat{\beta}^{(s+1,t)} &\;= \;\arg\max_\beta\;\; \sum_{i=1}^n \sum_{k=1}^K w_{ik+}^{(s,t)}\log \pi_k(x_i; \beta),\\
\hat{\mu}_k^{(s+1,t)}, (\hat{\tau}_k^2)^{(s+1,t)} &\;=\; \arg\max_{\mu_k, \tau_k^2}\;\; \sum_{i=1}^n \sum_{b=0}^1 w_{ikb}^{(s,t)} \log \phi(z_{i,b}; \mu_{k}, \tau_{k}^2 + \sigma_i^2).
    \end{split}
\end{align}

The update for $\beta$ is a standard optimization problem for an off-the-shelf likelihood-maximizing classifier, where we have one weighted ``observation'' for every combination of $x_i \in \cX$ and $\gamma_i \in [K]$, with case weight $w_{ik+}^{(s,t)}$. We provide five default methods in our package: \texttt{nnet::multinom}, \texttt{nnet::nnet}, \texttt{glmnet::glmnet}, \texttt{mgcv::gam}, and \texttt{VGAM::rrvglm}. In general, the analyst may use their favorite method, including more complex models such as deep neural networks, random forests, or gradient boosting. For the studies discussed above with one-dimensional covariates, we use a neural network model with one hidden layer and a natural cubic spline feature basis.

In the special case where all $\sigma_i^2$ are equal, the update for $\mu_k$ and $\tau_k^2$ amounts to calculating the weighted mean and variance of the $2n$ $z_{i,b}$ values for each component. In the general case, the update has no closed-form solution, so we use the \texttt{optim} package in R with the L-BFGS-B algorithm.

To calculate the $w_{ikb}^{(s,t)}$ correctly, we must take careful account of Jacobians in the nonlinear mappings $g(p)$ and $p(z)$; in particular, the slope of $g(p)$ is $\zeta$ times gentler in the ``blue'' region than it is elsewhere. Letting $dp$ represent an infinitesimal increment around some $m \in (0,\alphamax)$, we have
\begin{align*}
\mathbb{P}_\eta \left(b_i = b, m_i \in m \pm dp \mid \gamma_i, x_i\right)
&\;=\;  
\begin{cases} 
   \mathbb{P}_\eta \left(p_i \in m \pm dp \mid \gamma_i, x_i\right) & \text{if } b=0\\[2.0ex]
   \mathbb{P}_\eta \left(p_i \in \lambdamax- \zeta m \pm \zeta dp \mid \gamma_i, x_i\right) & \text{if } b=1
\end{cases}\\[2.0ex]
&\;=\; \frac{\phi(z_{i,b}(m); \mu_{\gamma_i}, \tau_{\gamma_i}^2+\sigma_i^2)}{\phi(z_{i,b}(m); 0, \sigma_i^2)} \cdot \zeta^b \cdot 2\,dp,
\end{align*}
where $z_{i,0}(m)=\sigma_i \Phi^{-1}(1-m)$ and $z_{i,1}(m) = \sigma_i \Phi^{-1}(1 - (\lambdamax -\zeta m))$ are the ``red'' and ``blue'' $z$-values whose $p$-values map to $m$, and the factor $\phi(z_{i,b}(m); 0,\sigma_i^2)$ in the denominator is the derivative of the one-sided $p$-value transform $p_i=1-\Phi(z_i/\sigma_i)$. 

Because $\sum_{k=1}^K \sum_{b=0}^1 w_{ikb}^{(s,t)} = 1$ for each $i$, we have for masked $p$-values:
\[
\mathbb{P}_\eta \left(\gamma_i=k, b_i = b \mid x_i, m_i\right) \;\propto\; \frac{\phi(z_{i,b}; \mu_{k}, \tau_k^2+\sigma_i^2)}{\phi(z_{i,b}; 0, \sigma_i^2)}  \cdot \zeta^b \pi_k(x_i; \beta).
\]
If we call the expression on the right-hand side $v_{ikb}$, then $w_{ikb} = v_{ikb}/\sum_{k',b'} v_{ik'b'}$ gives the correct probabilities. For already revealed $b_i$ values, $w_{ikb} = v_{ikb}/\sum_{k'} v_{ikb}$ for $b=b_i$ and $w_{ikb}=0$ otherwise.

\subsection{Tuning parameters and initialization}\label{subsec: Adaptive Masking Functions}

The EM method described above has a variety of tuning choices including the number of classes $K$, the feature basis for the classifier, and any tuning parameters for the classifier such as the number of hidden nodes in the neural network. To select values for the tuning parameters, we can fit a model on all of the masked data for every choice of tuning parameters and perform model selection with AIC, AICc, BIC, HIC, or cross-validation, with AIC as the default choice in our package.

The parameters $\mu_k, \tau^2_k$ are initialized using the $K$-means algorithm and $\beta$ is initialized with an intercept-only model. We also include an intercept-only model among the candidate featurizations by default to account for the possibility that the covariates are uninformative.

To operationalize our procedure, we also need to choose the masking function parameters $\alphamax$, $\lambdamin$, and $\lambdamax$. As we discussed in Section~\ref{sec: AdaPT_g}, there are several tradeoffs that the analyst should consider including their expectations about the number of rejections and the distribution of the null and alternative $p$-values.

Absent user input, our package makes default choices $\lambdamax = 0.9$ to exclude a possible $p$-value spike near 1, and chooses a common value $\alphamax = \lambdamin$ using a heuristic that maximizes $\alphamax$ subject to the constraint that the minimum number of rejections $R_{\min} = (\zeta\alpha)^{-1}$ is at least $\max\{1,n/300\}$:
\[
\frac{1}{\zeta\alpha} \geq \max\{1,n/300\} \iff \zeta \leq \left(\alpha \max \{1,n/300\}\right)^{-1}.
\]
This choice reflects a default expectation that the number of rejections we expect will scale roughly with $n$. We also expect heuristically that there are diminishing returns to increasing $\alphamax$ beyond 0.3 even if $n$ is very large. When $\lambdamax=0.9$ and $\alphamax=\lambdamin$, this corresponds to an upper bound $\zeta \geq 2$. Pulling these rules together, we set
\[
\alphamax = \lambdamin = \frac{0.9}{\zeta + 1}, \quad \text{ where } \zeta = \max\left\{2,\;\; \left(\alpha\max\{1,n/300\}\right)^{-1} \right\} = \max\left\{2, \;\; \min\left\{\frac{1}{\alpha}, \frac{300}{n\alpha}\right\}\right\}
\]

Table~\ref{table: default mask} shows the default choices for various settings of $n$, with $\alpha = 0.05$.

\begin{table}[h]
\centering
\caption{Default masking function parameters ($\alpha = 0.05$)}
\label{table: default mask}
\begin{tabular}{@{}ccccccc@{}}
\toprule
$n$ &  & $\alphamax$ & $\lambdamin$ & $\lambdamax$ &$\zeta$ & $R_{\min}$ \\ \cmidrule(r){1-1} \cmidrule(l){3-7} 
  $\le 300$ && $0.043$ & $0.043$ & $0.9$& $20$ & $1$\\
  $500$ && $0.069$ & $0.069$ & $0.9$& $12$ & $2$\\
  $1000$ && $0.13$ & $0.13$ & $0.9$&$6$ & $4$\\
  $\ge 3000$ && $0.3$ & $0.3$ & $0.9$&$2$ & $10$\\ \bottomrule
\end{tabular}
\end{table}

\section{Empirical comparison of \gmm with other methods}

\subsection{Review of methods under comparison}\label{sec: Methods under comparison}

In this section we compare \gmm with several other state-of-the-art methods for covariate-assisted multiple testing, which we describe below. To evaluate and compare these methods, we reproduce the empirical analysis of two earlier papers, \citet{Korthauer} and \citet{AdaFDR}, and provide new simulations of our own. Due to the large number of methods under comparison, and because some of the methods are either inapplicable to specific cases, or generically defined with indeterminate implementations for specific cases, we do not include every method in every study. In reproducing the studies from \citet{Korthauer} and \citet{AdaFDR}, we include the methods in each of the original papers as implemented there, along with \gmm. In addition, we include AdaPT-GLM$_g$, the AdaPT-GLM method with our new asymmetric masking function so that, where \gmm outperforms the original AdaPT-GLM, the reader may distinguish how much of the improvement is attributable to the new masking function and how much is attributable to the conditional GMM implementation.

The covariate-assisted methods fall into two groups: those that accept covariates in a generic predictor space, including in particular $\RR^d$, and those that accept a single categorical covariate representing membership in one of $G$ groups. While the group covariate could arise by binning continuous covariates, this approach does not generalize beyond one or two dimensions. 

Our analysis includes two methods that operate on groups: the {\em local false discovery rate} (LFDR) method of \citet{LFDR} and the {\em independent hypothesis weighting} (IHW) method of \citet{IHWignatiadis}. LFDR is an empirical Bayes method that estimates $\text{lfdr}(p_i \mid x_i)$ for every hypothesis in every group, and rejects as many hypotheses with small lfdr as possible, subject to the constraint that
\begin{equation}\label{eq: average FDR}
    \sum_{i: \; H_i \text{ rejected}} \widehat{\text{lfdr}}(p_i \mid x_i) \;\leq\; \alpha.
\end{equation}
The method is optimal when the covariate is categorical and lfdr is given by an oracle. In practice the lfdr is estimated using a Gaussian two-groups model, and the method's asymptotic FDR control guarantee relies on consistent estimation of the lfdr. When the categorical variable is the result of binning a continuous variable, it is challenging to choose the right number of bins; following \citet{Korthauer}, we use IHW's automatic binning method to select the bins.

IHW estimates hypothesis weights as a function of the covariate, by estimating the null proportion 
\[
\pi_0(x) = \PP(H_i \text{ is true } \mid x_i = x)
\]
in an empirical Bayes two-groups model. Then, IHW rejects $p$-values less than a hypothesis weighted threshold. While IHW is formally defined for a single categorical covariate / grouping variable, it can be applied to continuous covariates that have been binned. Like AdaPT, IHW achieves finite-sample FDR control by constructing the weights as functions of censored $p$-values, $p_i\mathbbm{1}\{p_i >\tau\}$ for $\tau \in(0,1]$, typically $\tau = 0.5$. However, this censoring destroys much of the information that is most useful for determining the empirical Bayes prior, since it completely obscures which regions of the predictor space have many very small $p$-values. By contrast, the masking scheme in AdaPT only partially obscures which $p$-values are small, and it is usually possible to impute the very small $p$-values accurately during the empirical Bayes estimation.

An interesting intermediate method is the {\em structure-adaptive Benjamini--Hochberg algorithm} (SABHA) of \citet{SABHA}. Like IHW, SABHA also uses censored p-values $p_i \mathbbm{1}\{p_i> \tau\}$ and estimates a weight $\hat{q}_i$ for each hypothesis, again interpretable as an estimator of the prior odds that $H_i$ is true. As with IHW, this censoring scheme limits the empirical Bayes estimation accuracy.

While \citet{SABHA} motivate their method in terms of structural relationships between the hypotheses that constrain the estimator $\hat{q} = (\hat{q}_1,\ldots,\hat{q}_n)$ such as group or ordinal structure, the structural information could represent locations in a predictor space. SABHA controls FDR in finite samples if the user applies an FDR level correction factor that is based on the Rademacher complexity of $\hat{q}_i$. Because \citet{SABHA} do not suggest implementations for their method that apply to generic predictor spaces, we do not include it in our comparisons.

The Boca--Leek (BL) method of \citet{BL} likewise estimates the null proportion $\pi_0(x)$ using logistic regression on $x$, using $Y_i = \mathbbm{1}\{p_i > \tau\}$ as a binary response variable. They estimate $\pi_0(x)$ for several values of $\tau$ and smooth the results over $\tau$ to obtain the final estimator. The BL method attains asymptotic FDR control provided the estimator for $\pi_0(x)$ is consistent.

{\em FDR regression} (FDRreg) of \citet{FDRreg} is another empirical Bayes method that models $\pi_0(x)$ as a logistic regression, but it is based on $z$-values instead of $p$-values and directly estimates an alternative density in addition to the null proportion. They use an EM scheme to estimate weights based on their method, and allow the user to specify a theoretical or empirical null \citep{efron_empirical}. \citet{FDRreg} do not prove theoretical FDR control guarantees, and \citet{Korthauer} find in their simulations that their method does not reliably control FDR at the advertised level.

The {\em black box FDR} (BB-FDR) method \citep{bbfdr} is yet another empirical Bayes method for FDR control, using a black box model (a deep neural network in their implementation) to fit a prior for the group probabilities in the two groups model. We find that the method can suffer major violations of FDR control under model misspecification, so we do not include BB-FDR in our experiments.

The method most similar to AdaPT is AdaFDR \citep{AdaFDR}, which attempts to directly learn an optimal rejection threshold surface from a parameterized family of candidate thresholds, specifically a linear combination of an exponential function plus several Gaussian bumps. AdaFDR uses the same mirroring technique as AdaPT to estimate the FDP for each threshold, but instead of masking $p$-values they use a cross-fitting scheme after splitting the data set. As a result, there is no finite-sample FDR control guarantee, but they do attain FDP control asymptotically. In our simulations, we find that AdaFDR occasionally violates FDR control, primarily for larger values of $\alpha$.

Finally, we compare the covariate-assisted methods to three methods that do not accept covariates: the well-known Benjamini--Hochberg \citep{BH} and Storey-BH \citep{Storey} methods and the adaptive shrinkage (ASH) method of \citet{ASH}, an empirical Bayes method that uses $z_i$ and $\sigma_i^2$ as input, and assumes the distribution of parameters $\theta_i$ is a mixture of a point mass at zero and $K$ Gaussian distributions centered at zero with predetermined variances. ASH differs from other methods by requiring the assumption that the effect sizes are unimodal. They fit $\pi_k$, the probability of each Gaussian component, by penalized maximum likelihood estimation, and use the fit distributions to estimate the lfdr, or to control the usual FDR by averaging as in \eqref{eq: average FDR}.

To facilitate comparisons with other methods, some of which do not allow for multiple covariates, all of the experiments in this section involve a single covariate. However we emphasize that a major advantage of AdaPT is its ability to incorporate many covariates at a time. Unlike methods that rely on asymptotic convergence of the empirical Bayes model parameters, AdaPT allows for multivariate or high-dimensional modeling of predictor variables. Table~\ref{table: FDR methods} systematically compares the methods as to their inputs, the nature of their FDR guarantees, and their general approach. 

While all methods except the BH procedure share the assumption that the $p$-values be independent of each other (positive dependence is enough for BH; see \citet{benjamini2001control}), several methods require additional assumptions. Like \gmm, AdaFDR requires that null $p$-values have a non-decreasing density. The original AdaPT requires the mirror-conservative assumption, which is weaker than non-decreasing density and holds under roughly the same sufficient conditions. The ASH method, which takes $z$-values as its inputs rather than $p$-values, requires that the $z$-values follow a unimodal distribution.

{

\begin{table}
\centering
\caption{Comparison of multiple testing methods used in empirical studies.}
\begin{threeparttable}
\begin{tabular}{llcll}
\toprule
  &Method & Inputs & FDR Guarantee & General Approach\\
  \midrule
  \multicolumn{5}{l}{\textbf{Generic covariates}}\\
  &\gmm & $z_i, \sigma_i^2$& Finite-sample&Est. optimal threshold\\

  &AdaPT & $p_i$ & Finite-sample& Est. optimal threshold\\

  &AdaFDR &$p_i$& Asymptotic FDP control 

  &Est. optimal threshold\\[5pt]

  &SABHA & $p_i$
  & Finite-sample\tnote{*} &Estimate  $\pi_0(x)$\\
  
  &BB-FDR & $p_i$& None &Estimate  $\pi_0(x)$\\
  
    &BL & $p_i$ & Asymptotic\tnote{\dag} & Estimate  $\pi_0(x)$\\
  
  &FDRreg & $z_i$ & None &Estimate  $\pi_0(x)$\\[10pt]

  \multicolumn{5}{l}{\textbf{Categorical covariate (groups)}}\\

   &IHW & $p_i$& Finite-sample\tnote{\ddag} &Estimate $\pi_0(x)$\\

  &LFDR & $p_i$& Asymptotic\tnote{\dag} &Estimate $\text{lfdr}(p \mid x)$\\[10pt]

  \multicolumn{5}{l}{\textbf{No covariates}}\\
  &BH & $p_i$ & Finite-sample & Est. const. threshold \\
  &Storey--BH & $p_i$ & Finite-sample & Est. const. threshold, $\pi_0$\\
  &ASH & $z_i,\sigma_i^2$& None& Estimate lfdr\\

  \bottomrule
\end{tabular}
\begin{tablenotes}
\item[*] For finite-sample FDR control, SABHA requires a correction based on the Rademacher complexity of the estimator for the null proportion.
\item[\dag] For their asymptotic FDR control guarantee, BL and LFDR require asymptotically consistent estimators for $\pi_0$ and $\text{lfdr}$ respectively. 
 \item[\ddag] By default, the R package for IHW implements an earlier version of the method that operates on uncensored $p$-values and controls FDR asymptotically.
\end{tablenotes}
\end{threeparttable}
\label{table: FDR methods}
\end{table}
}

\subsection{Logistic Simulations}\label{subsec: logistic simulations}

To empirically evaluate the methods described above, we simulate a scenario with a univariate predictor that, when large, indicates the likely presence of a non-null signal. Conditional on a covariate $x_i\sim\cN(0,1)$, we either sample the parameter of interest $\theta_i$ from a logistic distribution or set it equal to zero. The $(x_i, \theta_i, z_i)$ triples for $i=1,\ldots,n=3000$ are sampled independently from the model:
\begin{align*}
    x_i &\sim \cN(0,1)\\
    \gamma_i \mid x_i &\sim \text{Bern}(\pi_1(x_i)), \qquad \pi_1(x) = \frac{3}{4} \cdot \frac{e^{6x-9}}{1+e^{6x-9}}\\[5pt]
    \theta_i \mid x_i, \gamma_i &\sim \begin{cases}
       \mathrm{Logistic}\left(2,\frac{1}{2}\right) &\text{ if } \gamma_i = 1\\
       0 &\text{ if } \gamma_i = 0
    \end{cases}\\[5pt]
    z_i \mid x_i, \gamma_i, \theta_i &\sim \mathcal{N}(\theta_i,1).
\end{align*}

We chose a scaled logistic function for $\pi_1$, and the logistic distribution for the alternative density of $\theta_i$, so that none of the models are correctly specified. The parameters are chosen so that the covariate is very informative; the signal is strong enough to be detected by a method that makes good use of the covariate, but not strong enough for a method to make too many detections otherwise.

Under the same simulation setting, we consider three testing problems: (i) testing each point null $H_i:\;\theta_i = 0$ against the two-sided alternative $\theta_i \neq 0$, (ii) testing the one-sided null $H_i:\; \theta_i \leq 0$ against the one-sided alternative $\theta_i > 0$, and (iii) testing the interval null $H_i:\; |\theta_i| \leq 1$ against the two-sided alternative $|\theta_i| > 1$. In each case we calculate the standard $p$-value transform so that $p_i$ is uniform at the boundary of the null, but has strictly increasing density if $\theta_i$ is in the interior of the null. The $p$-value transform for the interval null is $p_i = 1-\Phi( |z_i| +1)+\Phi(-| z_i| + 1)$.

\begin{figure}
\centering
  \includegraphics[width=\columnwidth]{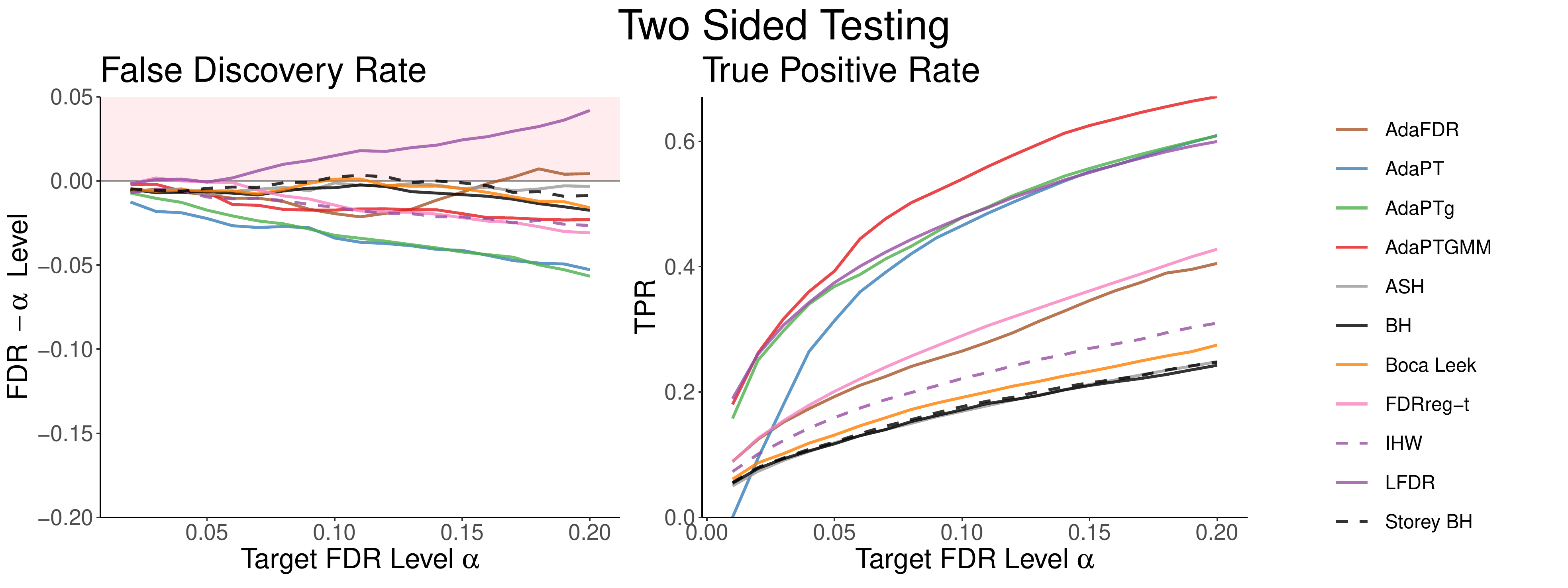}
\caption{FDR and TPR comparisons for testing the point null $H_i:\; \theta_i = 0$ in the logistic simulation. Each method is averaged over 100 initializations. Most of the covariate-assisted methods improve on the methods that do not use the covariates. The three variants of AdaPT achieve the highest power along with the LFDR method, which violates FDR control.}
\label{fig: logistic exp two sided}
\end{figure}
  
  \begin{figure}
\centering
  \includegraphics[width=\columnwidth]{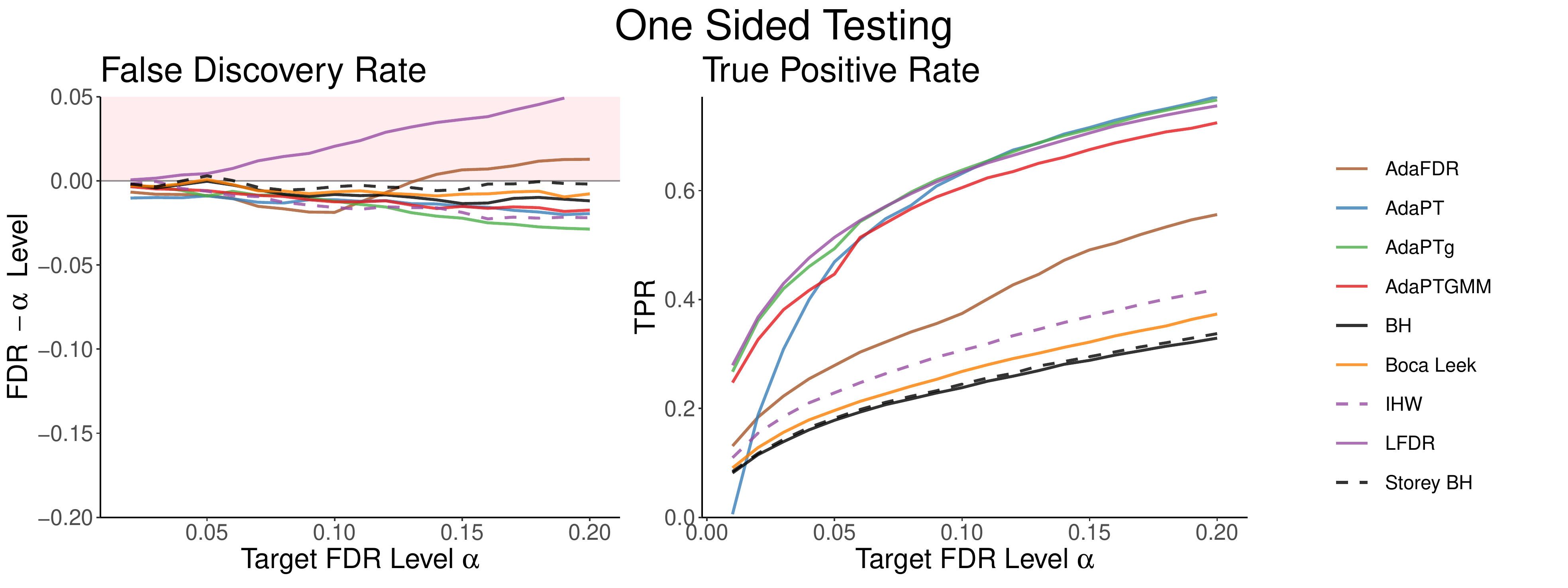}
\caption{FDR and TPR comparisons for testing the one-sided null $H_i:\; \theta_i \le 0$ in the logistic simulation. Each method is averaged over 100 initializations. The ordering between methods is similar to Figure~\ref{fig: logistic exp two sided}: all covariate-assisted methods improve on the methods that use no covariates, with the AdaPT methods improving the most along with LFDR, which does not control FDR. We also see that AdaFDR seems to be violating FDR control, due to solely asymptotic FDP guarantee.}
\label{fig: logistic exp one sided}
\end{figure}

\begin{figure}
\centering
  \includegraphics[width=\columnwidth]{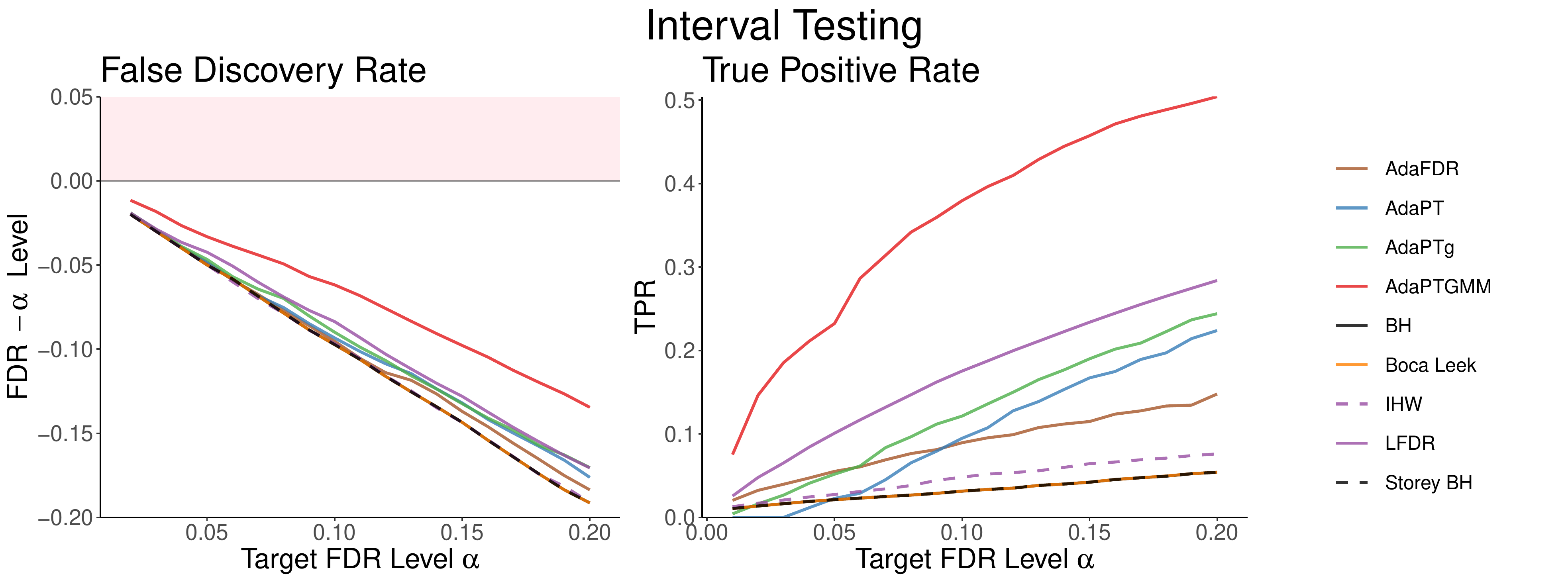}
\caption{FDR and TPR comparisons for testing the interval null $H_i:\; |\theta_i| \leq 1$ in the logistic simulation. Each method is averaged over 100 initializations. For this problem, \gmm achieves substantial gains over competing methods. AdaPT and \adaptg struggle the implementation models the distribution of null p-values as uniform.}
\label{fig: logistic exp interval}
\end{figure}

The three $p$-value distributions are shown in Figure~\ref{fig: logistic p-value dist} in Appendix~\ref{app: full case studies}. In particular, we see a super-uniform null p-value distribution for the interval null.

\subsection{Empirical studies from \citet{Korthauer}} \label{sec: Experiments}

\citet{Korthauer} evaluate most of the methods in Table~\ref{table: FDR methods} as to their performance on a wide range of simulation and real data experiments. We reproduce the simulation experiments in Appendix~\ref{app: full case studies}. In particular, they evaluate the methods' power on 32 settings involving different data sets and covariates. The data sets they study involve a wide variety of computational biology tasks including differential binding testing in ChIP-seq, gene set analysis (GSEA), genome-wide association testing (GWAS), differential abundance testing in microbiome data, bulk RNA-seq, and differential expression in single-cell RNA-seq (scRNA-seq). We discuss the data sets and covariates  in Appendix \ref{app: case study data sets} and refer the reader to \citet{Korthauer} for further details. We exclude LFDR, FDRreg, and ASH due to inconsistent FDR control mentioned in Section~\ref{subsec: FDR Simulations}. For transparency, we include a full heatmap with LFDR, FDRreg, and ASH in Appendix~\ref{app: full case studies}.

Figure~\ref{fig: Case Study experiments} is a heatmap with rows corresponding to FDR procedures and columns corresponding to a case study and covariate. Each square is colored from a gradient of white to dark blue, where darker squares correspond to more powerful methods. The most powerful method for each case study also has a white text label of the percentage of hypotheses rejected. We include \gmm as `adapt-gmm\_g' and the old implementation of AdaPT is labeled `adapt-glm'.

For the four experiments with test statistics and standard errors available (GWAS and RNA-seq data sets), we run \gmm using a two-sided null hypothesis and symmetric modeling assumption, as mentioned in Section~\ref{subsec: conditional gaussian mixture model}, using the standard error as an additional covariate.

In experiments with a small numbers of rejections, AdaPT can perform very poorly, a further example of shortcoming~\ref{shortcomings: min hypo}. Whereas previously AdaPT would make zero rejections, \gmm is able to perform well, even being the most powerful method in some situations. Overall, \gmm achieves the greatest power in $20$ out of $32$ case studies.

\begin{landscape}

\begin{figure}[h]
\centering
  \includegraphics[width=\columnwidth]{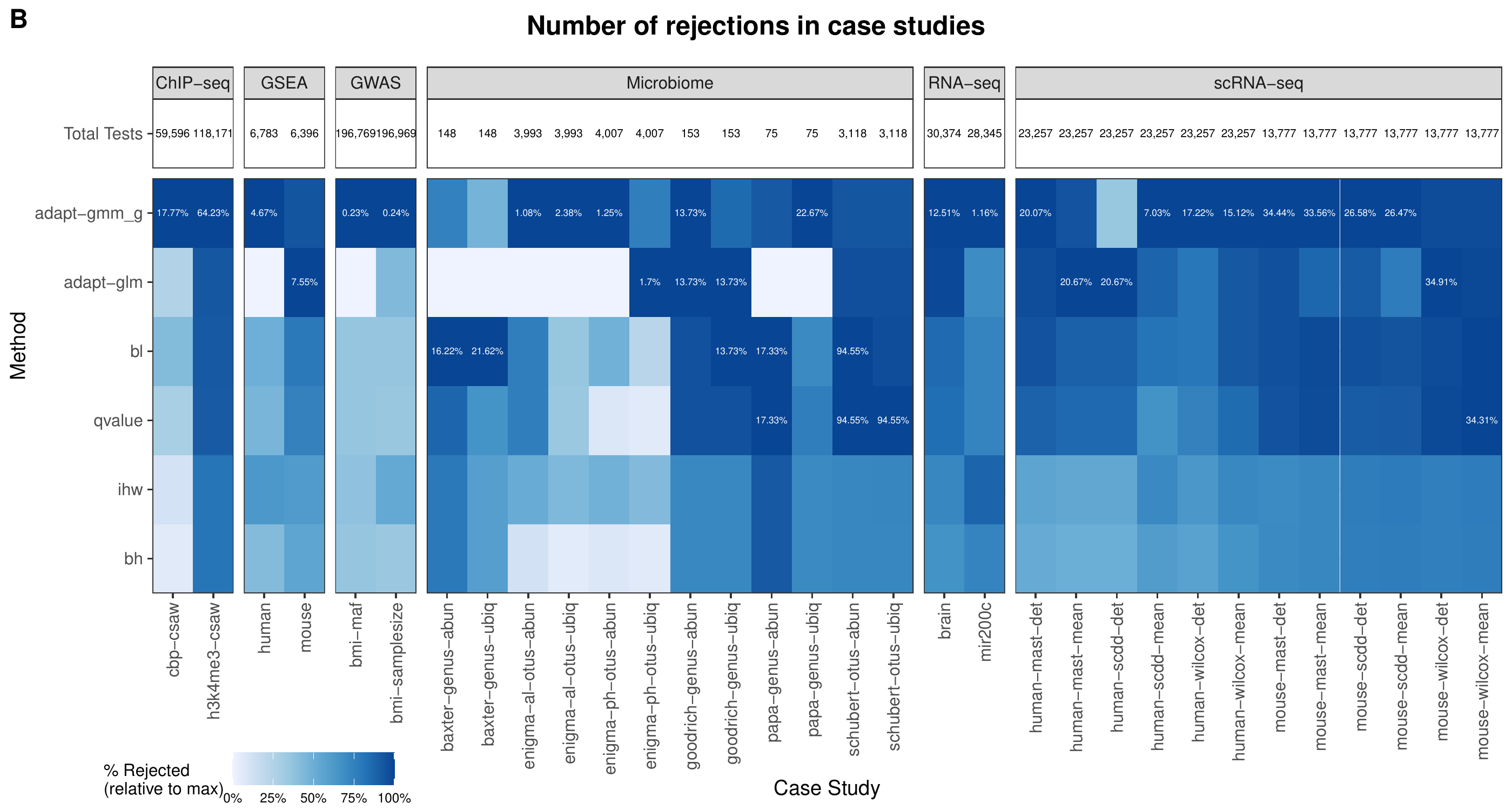}
  \caption{ Table of power in 32 case studies involving real data from a variety of experiments, compiled by \citet{Korthauer}. Each row corresponds to a multiple testing procedure and each column corresponds to a case study involving a data set and a single covariate. Darker squares correspond to more rejections and squares with white text correspond to the maximum rejection percentage within the column. This figure reproduces the main figure of \citet{Korthauer}, including \gmm and excluding LFDR, ASH, and FDRreg, which they found did not reliably control FDR in simulations. The full figure including LFDR, ASH, FDRreg, and AdaPT-GLM$_g$ can be found in Appendix~\ref{app: full case studies}. When $z$-values are unavailable for \gmm we construct them from the $p$-values via $z_i = \Phi^{-1}(1-p_i)$, where $\Phi$ is the standard Gaussian distribution function.}
\label{fig: Case Study experiments}
\end{figure}

\end{landscape}

\section{Discussion} \label{sec: Discussion}

\subsection{Extension to other parametric models}\label{sec:extensions}

While we implement \gmm as a method in the special case with Gaussian test statistics, it could be repurposed with minor modifications to other parametric models. In general, suppose $z_i \sim f_{\theta_i}(z)$ independently, for some parameter $\theta_i \in \RR$, that $p_i(z_i)$ is a $p$-value for testing some $H_i$ concerning the value of $\theta_i$, whose density is non-decreasing under $H_i$ so that \adaptg applies. Then by analogy to \gmm we could estimate a conditional mixture model for $\theta_i$ given $x_i$, using either Gaussian mixture components as in \gmm or replacing them with another distribution such as a conjugate prior for $\theta_i$. In either case, we can use the same EM framework as described in Section~\ref{sec: Implementation}, alternating between estimating the probability that each data point comes from each mixture component (the E step) and estimating the parameters of each mixture component along with a generic off-the-shelf classifier (the M step).

\subsection{Dependent $p$-values}

The greatest remaining weakness of \gmm (and also of AdaPT, and all the other competing methods discussed herein) is its assumption of independence across the $p$-values. This assumption is unrealistic in practice in many of the most common applications of multiple testing including GWAS, microarray studies, and fMRI studies, so relaxing it is an important direction for future research. \citet{DependenceFDRFithian} offers new technical tools for controlling FDR under dependence, including in settings with data-adaptive $p$-value weights, that may prove fruitful in relaxing the independence assumption. The adaptive knockoff method of \citet{ren2020knockoffs} suggests another possible way forward by incorporating side information into multiple testing in supervised learning problems.

\subsection{Summary}

The \adaptg method improves on the original AdaPT framework by generalizing its masking function. It inherits AdaPT's flexibility and its robust FDR control guarantee, and resolves the original method's two main performance issues: low power in small samples and null $p$-values close to 1. Our implementation, \gmm, is tailored to the common setting where effects are estimated with Gaussian errors, and is especially appropriate for testing composite null hypotheses such as one-sided or interval nulls. The method models the conditional distribution of $z$-statistics given covariates using a Gaussian mixture model with mixing proportions that depend on the covariates. Our EM estimation framework is compatible with any off-the-shelf method for modeling multinomial probabilities, and we implement it using a neural network with one hidden layer.

In reproduced experiments from \cite{Korthauer}, \cite{AdaFDR} and in new simulations, we find that \gmm and \adaptg is more powerful and reliable than AdaPT, and consistently delivers power that rivals or exceeds other state-of-the-art methods. We provide a package \texttt{AdaPTGMM} and provide user-friendly default parameters for ease of use.

Perhaps the greatest advantage of \gmm is that, like other AdaPT methods, its flexibility allows it to be extended in a great many ways depending on the problem specifics. Any classifier can be swapped in for our neural network, and the Gaussian mixture model can even be swapped out for any other conditional density estimation model, without threatening the finite-sample FDR guarantee.
\subsection*{Reproducibility}
The code to reproduce the experiments and simulations in this paper is publicly available at \url{https://github.com/patrickrchao/AdaPTGMM_Experiments}.
\section*{Acknowledgments}

William Fithian is partially supported by the NSF DMS-1916220 and a Hellman Fellowship
from Berkeley. We are grateful to Lihua Lei, Jelle Goeman, Aaditya Ramdas, Boyan Duan, Patrick Kimes, Ronald Yurko, Max Grazier-G'Sell, Nikos Ignatiadis, and Kathryn Roeder for insights we have gleaned from helpful conversations with them, and to Patrick Kimes for his help in reproducing the experiments from \citet{Korthauer}.

\clearpage

\bibliography{references}  
\clearpage
\appendix
\section{Appendix}
\subsection{FDR Control} \label{app: FDR control}
First, define the following variables. 
\begin{align*}
  C_t&=\{i\in\mathcal{H}_0 \text{ and } i \in \cM_t \}\\
  b_i&=\mathbbm{1}\{\lambdamin\le p_i\le \lambdamax\}\\
  U_t&=\sum_{i\in C_t} b_i\\
  V_t &=\sum_{i\in C_t} 1-b_i=|C_t| - U_t \numberthis \label{eq: V_t}
\end{align*}
Following from the proof from \cite{AdaPT}, define the filtration $\mathcal{F}$ for $t=0,1,\ldots,$:
\begin{align*}
  \mathcal{F}_t&=\sigma((x_i,m_{i})_{i=1}^n,A_t,R_t,\cM_t,\{b_i:i\in \cM_t^\setcomp\})\\
  \mathcal{F}_{-1}&=\sigma((x_i,m_{i})_{i=1}^n,\cM_0,\{b_i:i\in \cM_0^\setcomp\})
\end{align*}
and lastly define the $\sigma$-fields
\begin{align*}
  \mathcal{G}_{-1}&=\sigma((x_i,q_i)_{i=1}^n,(b_i)_{i\not\in \mathcal{H}_0})\\
  \mathcal{G}_{t}&=\sigma(\mathcal{G}_{-1},C_t,(b_i)_{i\not\in C_t},U_t).
\end{align*}

These imply that given $b_i$ and $m_i$, we can recover $p_i$.
\begin{align*}
p_i&=\mathbbm{1}\{b_i=1\}(\lambdamin+(\alpha-m_i)\zeta\}+\mathbbm{1}\{b_i=0\}m_i.
\end{align*}

\begin{lemma}\label{lemma: probability of big}
Assume that the null $p$-values are mutually independent and independent from the non-null $p$-values, and assume the null $p$-values have non-decreasing density. Then for $i\in\mathcal{H}_0$,
\begin{align}
  \PP(b_i=1\mid \mathcal{G}_{-1})\ge\frac{\zeta}{1+\zeta}.
\end{align}
\end{lemma}
\begin{proof}
Since we know that since the density of null $p$-values is non-decreasing,
\begin{align}
  \zeta\PP(p_i\in[0,\alphamax]) \le \PP(p_i\in[0,\alphamax\zeta])\le \PP(p_i\in[\lambdamin,\lambdamax]). \label{eq: lemma density non-decreasing}
\end{align}
From the definition of $\PP[b_i=1\mid \mathcal{G}_{-1}]$,
\begin{align*}
  \PP(b_i=1\mid \mathcal{G}_{-1}) &= \frac{\PP(p_i\in[\lambdamin,\lambdamax]\mid \mathcal{G}_{-1}]}{\PP(p_i\in[\lambdamin,\lambdamax) \mid\mathcal{G}_{-1}]+\PP(p_i\in[0,\alphamax]\mid \mathcal{G}_{-1})}\\
  &\ge 
  \frac{\PP(p_i\in[\lambdamin,\lambdamax]\mid \mathcal{G}_{-1})}{\PP(p_i\in[\lambdamin,\lambdamax] \mid\mathcal{G}_{-1})+\PP(p_i\in[\lambdamin,\lambdamax]\mid \mathcal{G}_{-1})/\zeta}\\
  &=\frac{\zeta}{\zeta+1}.
\end{align*}
\end{proof}

\begin{proof}[Proof of \textbf{Theorem~\ref{theorem: fdr control}}]

Let $\hat t$ be the stopping time when 
\[\widehat{\text{FDP}}_{\hat t} = \frac{1+A_{\hat t}}{\zeta R_{\hat t}} \le \alpha.\]
Let $V_t$ and $U_t$ correspond to the number of null $p$-values that contribute to $R_t$ and $A_t$ respectively, defined explicitly in equation~\eqref{eq: V_t}.
\begin{align}
\text{FDP}_{\hat t} = \frac{V_{\hat t}}{R_{\hat t}\vee 1} = \frac{(1+U_{\hat t})/\zeta}{R_{\hat t}\vee 1}\frac{V_{\hat t}}{(1+U_{\hat t})/\zeta}\le \frac{(1+A_{\hat t})/\zeta}{R_{\hat t}\vee 1}\frac{V_{\hat t}}{(1+U_{\hat t})/\zeta}=\alpha\zeta\left(\frac{V_{\hat t}}{1+U_{\hat t}}\right).
\end{align}
This follows from $U_{\hat t}\le A_{\hat t}$. We would like to show that $\E[V_{\hat t}/(1+U_{\hat t})]$ can be upper bounded by $\frac{1}{\zeta}$.

By lemma~\ref{lemma: probability of big},
\[\rho_i=\PP(b_i=1\mid \mathcal{G}_{-1})\ge\frac{\zeta}{1+\zeta}.\]
for all $i\in\mathcal{H}_0$. Furthermore, $\mathcal{F}_t\subseteq \mathcal{G}_t$ since
\begin{align*}
  A_t&=U_t + |\{ i\not\in \mathcal{H}_0:\lambdamin\le p_i\le \lambdamax \text{ and } i \in \mathcal{M}_t \} |\\
  R_t&=V_t + |\{ i\not\in \mathcal{H}_0:0\le p_i\le \alphamax \text{ and } i \in \mathcal{M}_t \} |.
\end{align*}
By applying the lemma 2 from AdaPT \cite{AdaPT},
\begin{align*}
  \E[\text{FDP}_{\hat t}\mid \mathcal{G}_{-1}]&\le \alpha\zeta\E\left[\frac{V_{\hat t}}{1+U_{\hat t}}\mid \mathcal{G}_{-1}\right]=\alpha\zeta\E\left[\frac{1+| C_{\hat t}|}{1+U_{\hat t}}-1\mid \mathcal{G}_{-1}\right]\\
  &\le \alpha\zeta\left(\left(\frac{\zeta}{1+\zeta}\right)^{-1}-1\right)=\alpha.
\end{align*}
By taking another expectation and by the law of iterated expectation, \[\FDR = \E[\text{FDP}]\le\alpha.\]
\end{proof}

\subsection{Optimality of Revealing Procedure}
\begin{theorem}\label{theorem: game}
 Consider a game with $n$ cards $c_1,c_2,\ldots,c_n$, each with one side white and the other side either red or blue. On the white side of each card is printed a number $q_i \in [0,1]$ reflecting the probability the other side is blue, so that 
 \[
 B_i = 1\{c_i \text{ is blue}\} \simind \text{Bern}(q_i).
 \]
 The game begins with all cards facing white-side-up, and the player being told how many of the cards are red and how many blue (i.e., $S_0=\sum_i B_i$ is revealed). On each turn $t=1,\ldots,n$, the player chooses one card, flips it over, observes its color, and removes it from the table. The game ends the first time $S_t$, the number of remaining face-down blue cards at the end of the turn, falls below a fixed threshold sequence $s_t$, and the player wants to end the game as soon as possible; i.e., the player wants to minimize $\tau = \min \{t:\; S_t \leq s_t\}$.
 
 Assume without loss of generality that the cards are ordered with $q_1 \geq q_2 \geq \cdots \geq q_n$. Then, regardless of the values taken by $q$ or $s$, or the initial condition $S_0 = \sum_i B_i$, it is optimal for the player to reveal the cards in order $c_1,c_2,\ldots,c_n$.
\end{theorem}

We note that this theorem states there is an optimal \textit{fixed} policy, meaning all choices are determined by the starting state, in contrast to an {\em adaptive} policy whose choices depend on already-revealed cards.

\begin{proof}

  We proceed by induction. Let $\mathcal{M}_t \subseteq [n]$ denote the indices of face-down (masked) cards at time $t = 1,\ldots,n$, when choosing the $t$th card to reveal. If 
  \[    
  W_t = \left(B_{\mathcal{M}_t^\setcomp}, S_{t-1}\right), \quad \text{ where } S_{t-1} = \sum_{i \in \mathcal{M}_t} B_i,
  \]
  denotes the information available to the player at time $t$, then a policy $f = \left(f_1(W_1),\ldots,f_{n}(W_{n})\right)$ defines the player's choice of card at each step as a function of the available information. Our claim is that the fixed policy $f^*$ with $f_t^* = c_t$ almost surely is optimal.

  \textbf{Base Case:}\\
  If $n=1$ there is only one possible strategy, which always chooses $c_1$.

  \textbf{Inductive Step:}\\
  First, note that the player has no information to adapt to on turn $t=1$ except the initial condition $S_0$. If the player selects $c_i$, then $B_i$ is revealed and $S_1 = S_0-B_i$. Either $S_1 \leq s_1$ and the game stops immediately with $\tau = 1$, or the player is playing a new version of the game with $n-1$ cards, probabilities $q_1,\ldots,q_{i-1},q_{i+1},\ldots,q_n$, target thresholds $s_2,\ldots,s_n$, minimization target $\tau - 1$, and with the initial condition that there are $S_1$ face-down blue cards. Further, the conditional distribution of $B_{\mathcal{M}_1}$ given $W_1$ is identical to the card color distribution in the new game.
  
  Applying the inductive hypothesis, then, the fixed policy $f^{(i)} = (c_i, c_1, \ldots, c_{i-1},c_{i+1},\ldots,c_n)$ is optimal among all strategies $f$ beginning with $f_1 = c_i$. Our goal is to show that $f^* = f^{(1)}$ is at least as good as $f^{(i)}$ for $i > 1$.
  
  \newcommand{\tfi}{\tilde{f}^{(i)}}
  For $i>1$, consider the fixed policy $\tfi = (c_1,c_i,c_2,\ldots,c_{i-1},c_{i+1},\ldots,c_n)$ which is $f^{(i)}$ with the first two steps swapped. Note $f^{(i)}$ and $\tfi$ result in the same $S_t$ for all $t\neq 1$, and they only differ on $S_1$ if exactly one of $c_1$ and $c_i$ is blue. The only question to ask in comparing these policies is which of them is more likely to reveal the blue one first in this case. As a result, $\tfi$ dominates $f^{(i)}$ provided that
  \[
  \PP\left[B_1 = 1, B_i = 0 \mid S_0=s\right] \geq \PP\left[B_i = 1, B_1 = 0 \mid S_0=s\right],
  \]
  for $1 \leq s \leq n-1$ (otherwise the policies produce identical trajectories).  If $q_1 = 1$ or $q_i=0$, the right-hand probability is 0; otherwise
  \[
  \frac{\PP\left[B_1 = 1, B_i = 0 \mid S_0 = s\right]}{\PP\left[B_i = 1, B_i = 0 \mid S_0 = s\right]} 
    \;=\; \frac{\PP\left[B_1=1, B_i=0,\sum_{j\neq 1,i}^n B_j =s-1\right]}{\PP\left[B_i=1, B_1=0,\sum_{j\neq 1,i}^n B_j =s-1\right]}
    \;=\; \frac{q_1(1-q_i)}{q_i(1-q_1)} \;\geq\; 1,
  \]
  so $S_1\leq s_1$ is at least as likely under $\tfi$, and $S_2,\ldots,S_n$ coincide for the two strategies almost surely. As a result, $\tfi$ is at least as good as $f^{(i)}$, and is also no better than $f^* = f^{(1)}$, which is optimal among policies with $f_1 = c_1$.
  
\end{proof}
\begin{proof}[Proof of \textbf{Theorem~\ref{theorem: optimal revealing with oracle}}]
\label{pf: optimal revealing with oracle}
  We utilize the earlier Theorem~\ref{theorem: game} by selecting 
  $s_{t}=\frac{\alpha\zeta (\vert \cM_0\vert- t) - 1}{1+\alpha\zeta}.$ In the presentation of the previous theorem, we defined $S_t$ as the number of blue cards remaining, we may redefine $S_t$ as $A_t$ in the formulation of AdaPT, the number of masked hypotheses in the blue region.
  \begin{align*}
      A_{\hat t}\le s_{\hat t}\;&=\;\frac{\alpha\zeta (\vert \cM_0\vert- \hat t) - 1}{1+\alpha\zeta}\\
      A_{\hat t} + 1\;&\le\; \alpha\zeta (n- \hat t-A_{\hat t})\\
         \frac{(A_{\hat t}+1)/\zeta}{R_{\hat t}}\;& \le\; \alpha,
  \end{align*}
  since $A_t+R_t=\vert \cM_0\vert-t$. Therefore by selecting $s_{t}$, the most powerful procedure is to reveal the masked hypotheses in decreasing order of $q_i$, or reveal hypotheses most likely to be blue.
\end{proof}

\subsection{Masking Function Shapes}\label{app: masking function shapes}
\begin{figure}[t]
\centering
\begin{subfigure}{.49\textwidth}
  \centering
  \includegraphics[width=0.9\columnwidth]{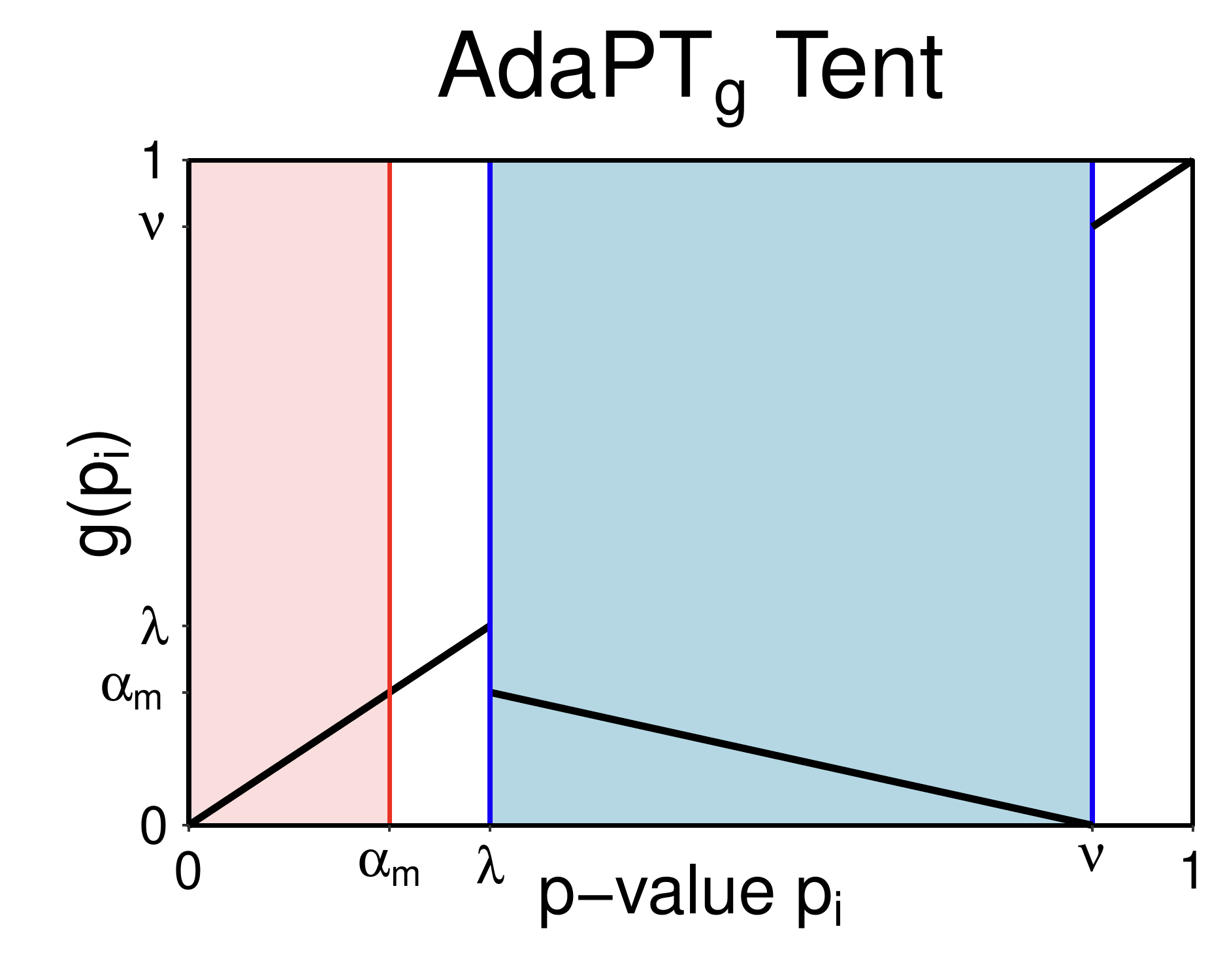}
\end{subfigure}\hfill
\begin{subfigure}{.49\textwidth}
  \centering
  \includegraphics[width=0.9\columnwidth]{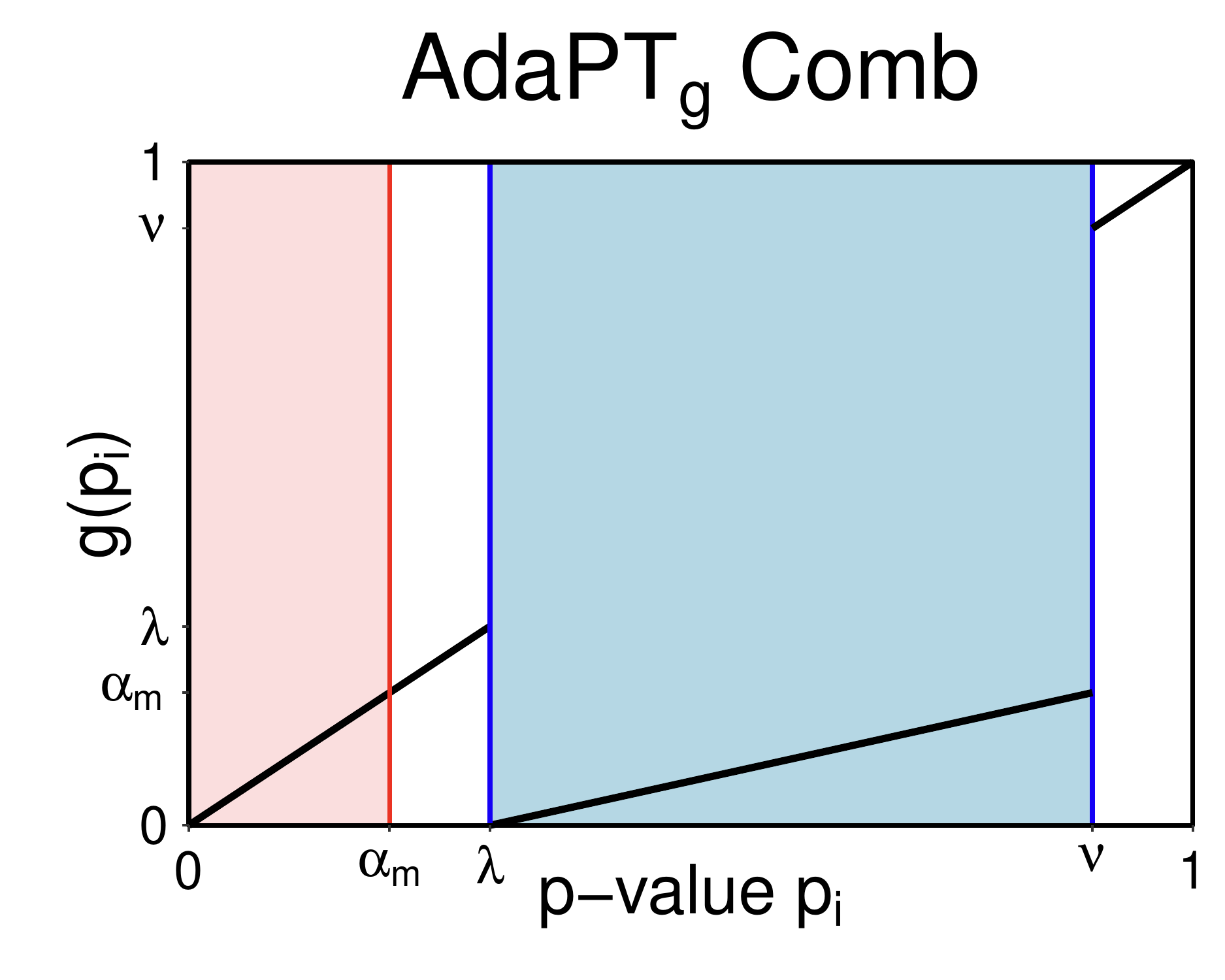}
\end{subfigure}
\caption{Masking function examples. The colored regions correspond to $p$-values that are masked and the y-axis are the masked $p$-values, what the analyst would observe. Left: Tent masking function, Right: Comb masking function.}
\label{fig:Tent Comb Masking Functions}
\end{figure}
The masking function $g(p)$ maps the region $[0,\alpha_m]$ to $[\lambda,\nu]$. Given a linear mapping, there is flexibility in which direction to map, namely $g(0)=g(\nu)$ or $g(0)=g(\lambda)$. We denote the first mapping as `tent' shaped and the second as `comb' shaped, where the naming convention follows from the shape of the red and blue regions in Figure~\ref{fig:Tent Comb Masking Functions}.
\begin{equation*}
    g_{\text{tent}}(p) = \begin{cases} 
    (\lambdamax - p)/\zeta & p \in [\lambdamin , \lambdamax]\\
    p & \text{otherwise}
    \end{cases},\qquad  
    g_{\text{comb}}(p) = \begin{cases} 
    (p-\lambda)/\zeta & p \in [\lambdamin , \lambdamax]\\
    p & \text{otherwise}
    \end{cases}
\end{equation*}
The main consideration of whether to use a tent or comb masking function boils depends on the density of $p$-values. In particular, since we would like to differentiate between alternative $p$-values near $0$ and $g^{-1}(0)$, we would like to map the region with lowest density to $0$.

For many situations, the tent masking function is sufficient. If the alternative $p$-value inflates the density at $\lambda$, then the tent masking function will perform better. However if the density of null $p$-values is super uniform, i.e. under interval nulls, then the null $p$-value density at $\nu$ may be substantially greater than at $\lambda$, meaning a comb masking function would be superior.

To account for these differences, we set the default masking shape to be the tent shape, and for interval testing we choose the comb shape.
\section{Optimization Details}
\subsection{Weighting terms $w_{ikb}$}\label{app: weighting terms w ikb}
The probability $w_{ikb}$ may be split into two cases, when $p_i$ is known and unknown.

\textbf{Unknown:}\\
\begin{align*}
  w_{ikb}&=\PP[\gamma_i=k,b_i=b \mid x_i, m_i]\\
  &=\frac{\PP[\gamma_i=k,b_i=b,m_i\mid x_i]}{\PP[m_i\mid x_i]}\\
  &=\frac{\PP[\gamma_i=k\mid x_i]\PP[b_i=b,m_i\mid \gamma_i=k]}{\sum_{k',b'}\PP[\gamma_i=k'\mid x_i]\PP[b_i=b',m_i\mid \gamma_i=k']}. \numberthis \label{eq: w ika computation}
\end{align*}

\textbf{Known:}\\
\begin{align*}
\PP[\gamma_i=k\mid x_i,p_i]&=\PP[\gamma_i=k\mid x_i,p_i]\\
  &=\frac{\PP[\gamma_i=k, p_i\mid x_i]}{\PP[p_i\mid x_i]}\\
  &=\frac{\PP[\gamma_i=k\mid x_i]\PP[ p_i\mid \gamma_i=k]}{\sum_{k'}\PP[\gamma_i=k'\mid x_i]\PP[p_i\mid \gamma_i=k']}. \numberthis 
\end{align*}
The terms $\PP[b_i=b,m_i\mid \gamma_i=k]$ are evaluated in Appendix~\ref{app: one sided p-value Probability conditioned on class}, the terms $\PP[\gamma_i=k\mid x_i]$ follow from the fitted probabilities of the multinomial model.

\subsection{One Sided: $p$-value Probability Conditioned on Class}\label{app: one sided p-value Probability conditioned on class}
Let $\Phi(z)$ be the CDF of a standard normal up to $z$. In equation~\eqref{eq: w ika computation}, we would like to evaluate
\[\PP[b_i=b,m_i\mid \gamma_i=k].\]
For simplicity, we may consider $b=1$ and testing one-sided null hypotheses,
\begin{align} \label{eq: one sided p-value}
p_i=q(z_i)=1-\Phi\left(\frac{z_i}{\sigma_i}\right).
\end{align}
Let $f_k$ be the density for $z$ given $\gamma=k$.
\begin{align*}
  \PP[b_i=1,m_i\in (p_i\pm dp) \mid \gamma =k]&=\PP[p_i\in g^{-1}(p_i\pm dp),b_i=1 \mid \gamma = k]\\
  &=\PP[p_i\in (p_{i,1}\pm \zeta dp) \mid \gamma = k]\\
  &= \zeta dp f_k( q^{-1}(p_{i,1})) \frac{d}{dp} \left[q^{-1}( p_{i,1})\right]\\
  &=\zeta dp \frac{f_k(z_{i,1})}{ | q'(q^{-1}( p_{i,1}))| }\\
  &=\zeta dp \frac{f_k(z_{i,1})}{| q'(z_{i,1})|}
\end{align*} 

Since $z\sim \mathcal{N}(\theta,\sigma_i^2)$, $f_k(z_{i,1})=\phi(z_{i,1}; \mu_k,\tau_k^2+\sigma_i^2)$.
\begin{align}
  \PP[b_i=b,m_i\mid \gamma =k] &\propto \frac{\phi( z_{i,b}; \mu_k,\tau_k^2+\sigma_i^2) \zeta^b}{\phi( z_{i,b}; 0,\sigma_i^2)}. \label{eq: probability p-value}
\end{align}

\section{Testing Point Null} \label{app: two sided testing}
Consider the two sided null hypothesis $H_i=\theta_i =0$ with alternative $\theta_i\neq 0$.
\[H_0: \theta =0, \quad H_1:\theta \neq 0.\]
Then we have
\begin{align} \label{eq: two sided p-value}
p_i=q(z_i)=2\left(1-\Phi\left(\left| \frac{z_i}{\sigma_i} \right|\right)\right) .
\end{align}

As mentioned in Section~\ref{subsec: Masking Function Families}, there are two pieces of missing information, $\sgn(z_i)$ and $b_i$, with a total of four possibilities. We choose to reveal we reveal $\sgn(z_i)(-1)^{b_i}$, leaving two possibilities. We choose to reveal $\sgn(z_i)(-1)^{b_i}$ rather than $\sgn(z_i)$ as it increases the distance between the candidate values of $z_i$. We may retain our previous notation of $p_{i,b}$ and $z_{i,b}$ corresponding to which of the two possibilities has $b_i=b$.

Now evaluating equation~\eqref{eq: w ika computation} and assuming $b_i=1$,
\begin{align*}
  \PP[b_i=1,m_i\in (p_i\pm dp) \mid \gamma =k]&=\PP[p_i\in g^{-1}(p_i\pm dp),b_i=1 \mid \gamma = k]\\
  &=\PP[p_i\in (p_{i,1}\pm \zeta dp) \mid \gamma = k]\\
  &= \zeta dp f_k( q^{-1}(p_{i,1})) \frac{d}{dp} \left[q^{-1}( p_{i,1})\right]\\
  &=\zeta dp \frac{f_k(z_{i,1})}{ | q'(q^{-1}( p_{i,1}))| }\\
  &=\zeta dp \frac{f_k(z_{i,1})}{| q'(z_{i,1})|}
\end{align*} 

This gives us
\begin{align}
  \PP[b_i=b,m_i\mid \gamma =k] &\propto \frac{\phi( z_{i,b} ; \mu_k,\tau_k^2+\sigma_i^2) \zeta^b}{2\phi( | z_{i,b}| ; 0,\sigma_i^2)}. \label{eq: two sided probability p-value}
\end{align}

\section{Interval Null Testing}\label{app: interval null testing}

We also provide functionality for testing interval null hypotheses,
\[H_i: | \theta_i| \le \delta.\]
Given observed z-scores $z$, the test statistic is $| z |$. We may evaluate the $p$-value from the property that Gaussian distributions are symmetric about the mean and are location scale families,
\begin{align}\label{eq: interval p-value}
  p_i=q_\delta(z_i)=1-\Phi\left (\left| \frac{z_i}{\sigma_i}\right| +\delta\right) + \Phi\left(-\left| \frac{z_i}{\sigma_i}\right| +\delta\right).
\end{align}

Similarly to the two-sided case, there are total four possibilities, from the sign of $z_i$ and $b_i$. Since the distribution of $z$ is not necessarily symmetric under the null, we cannot reveal $\sgn(z_i)(-1)^{b_i}$ as we did when testing two sided null hypotheses. Therefore we have four possible unknown values to impute.

Define $b_i'=\mathbbm{1}\{z_i>0\}$. To compute $w_{ikb}$ in equation~\eqref{eq: w ika computation}, the sum is now over $4$ total values of $b_i$ and $b_i'$. 

Similar to Appendix~\ref{app: one sided p-value Probability conditioned on class}, we would like to evaluate
\[\PP[b_i=b,b_i'=b',m_i\mid \gamma_i=k]\]
for the interval null case. As an example, assume $b=1$ and $b'=0$. 
\begin{align*}
  \PP[b_i=1,b_i'=0,m_i\in (p_i\pm dp) \mid \gamma =k]&=\PP[p_i\in (g^{-1}(p_i\pm dp),b_i=1,b_i'=0, \mid \gamma = k]\\
  &=\PP[p_i\in (p_{i,b}\pm \zeta dp),z_i\le 0 \mid \gamma = k]\\
  &=\zeta dp f_k(-q_r^{-1}(p_{i,b})) \frac{d}{dp} \left[-q_r^{-1}( p_{i,b})\right]\\
  &=\zeta dp\frac{f_k(-z_{i,b})}{ | q_r'(-q_r^{-1}( p_{i,b}))| }\\
  &=\zeta dp\frac{f_k(-z_{i,b})}{| q_r'(-z_{i,b})| }
\end{align*}

Taking the derivative of equation~\eqref{eq: interval p-value},
\begin{align}\label{eq: interval null probability p-value}
  \PP[b_i=b,b_i'=b',m_i\mid \gamma =k] &\propto \frac{\phi( z_{i,b}; \mu_k,\tau_k^2+\sigma_i^2) \zeta^{\mathbbm{1}\{b=1\}}}{\left(\phi( -| z_{i,b} | +r; 0,\sigma_i^2)-\phi(| z_{i,b}| +r; 0,\sigma_i^2)\right)}.
\end{align} 

\section{Case Studies from \citet{Korthauer}}\label{app: case study data sets}
\citet{Korthauer} compare empirical power on a wide variety of computational biology data sets, using $\alpha=0.05$ for all experiments. For further details and sources for the data, we direct the reader to \cite{Korthauer} and their additional files.

\textbf{ChIP-seq:}\\
Two chromatin immunoprecipitation sequencing (ChIP-seq) data sets were used, testing differential binding analyses between cell lines. The covariate is the mean read depth, the average coverage for the region. 

\textbf{GSEA:}\\
\cite{Korthauer} applied gene set enrichment analysis (GSEA) on two RNA-seq data sets, testing changes in gene expression. The covariate is the size of the gene set.

\textbf{GWAS:}\\
For the genome-wide association studies, the covariates are the sample size of the variant and minor allele frequency, the frequency of the second most common allele in a population.

\textbf{Microbiome:}\\
\cite{Korthauer} explore a variety of differential abundance analyses and correlation analyses. The differential abundance analyses use a Wilcoxon rank sum test on the opterational taxonomic units (OTUs), whereas the correlation analyses use a Spearman correlation test between OTUs and pH, Al, and SO$_4$, with the ubiquity (percentage of samples with the OTU) and mean nonzero abundance as covariates. 

\textbf{RNA-seq:}\\
The two RNA-seq data sets are samples from the GTEx project and an experiment with microRNA mir200c, testing differential expression with the mean gene expression level as covariate. 

\textbf{scRNA-seq:}\\
Lastly, the two single cell RNA-seq data sets use three methods for differential expression analyses, scDD \citep{scDD}, MAST \citep{MAST}, and the Wilcoxon rank-sum test. These three methods were used in conjunction with two covariates, the mean nonzero gene expression and detection rate of the gene.

\section{Empirical studies from \citet{AdaFDR}} \label{subsec: AdaFDR Experiments}
\citet{AdaFDR} also evaluate the performance of several of the methods in Table~\ref{table: FDR methods}, in the paper where they propose AdaFDR. We reproduce these case studies, adding \gmm and AdaPT-GLM$_g$. The experimental results are shown in Table~\ref{table: AdaFDR Experiments} and employ the following data sets.

\textbf{GTEx:}\\
The genotype-tissue expression (GTEx) data set comprises of expression quantitative trait loci (eQTL) (\cite{GTEx}), testing the association between a single nucleotide polymorphisms (SNP) and an eQTL. \cite{AdaFDR} use two sets of adipose tissue, adipose subcutaneous and adipose visceral omentum tissue. The four covariates for these experiments are the distance between SNPs and the gene transcription start site, the log gene expression level, the alternative allele frequency, and the chromatin state of the SNP. 
For computation purposes, the authors use a smaller version of the full data set comprising of 300,000 associations. The authors chose an FDR level of $\alpha=0.01$.

\textbf{RNA-Seq:}\\
The second set are RNA-Seq data sets, specifically the Bottomly \citep{Bottomly}, Pasilla \citep{Pasilla}, and Airway \citep{Airway} data sets, also used by \citep{AdaPT,IHWignatiadis}. We are interested in estimating differential expression significance, with the logarithm of the normalized counts as the single covariate. The authors choose an FDR level of $\alpha=0.1$.

\textbf{Microbiome:}\\
The microbiome experiments (\cite{Microbiome}) are the same as the enigma microbiome experiments from \cite{Korthauer}, except both covariates are used together. For these experiments $n\approx 4000$, testing correlation between operational taxonomic units (OTUs) and the pH and Al. The covariates are the ubiquity (percentage of samples with the OTU) and mean nonzero abundance. The authors chose an FDR level of $\alpha=0.2$.

\textbf{Proteomics:}\\
The proteomics data set is consists of comparing yeast cells treated with rapamycin and dimethyl sulfoxide, testing the differential abundance of $n=2666$ proteins using Welch's t-test, with the number of peptides as the covariate (\cite{Proteomics}). The authors chose an FDR level of $\alpha=0.1$.

\textbf{fMRI:}\\
The two functional magnetic resonance imagining (fMRI) data sets test response to stimulus for voxels in the human brain (\cite{fMRI}). In the auditory data set, a participant is given auditory stimulus, and in the imagination data set, the participant is tasked to imagine playing tennis. The four covariates are the categorical Brodmann area label (\cite{Brodmann}) and the spatial $(x,y,z)$ location of the voxel. The authors chose an FDR level of $\alpha=0.1$.
{
\renewcommand{\arraystretch}{1.5}
  \begin{table}[h]
    \begin{threeparttable}
  \small
  \caption{AdaFDR Experiments}
  \label{table: AdaFDR Experiments}
  \centering
    \begin{tabular}{l c c c c c c H c c }
    \toprule
    Data set  & BH & SBH& AdaPT & IHW & BL & AdaFDR &\gmm Old & \adaptg &\gmm\\
      \midrule
      GTEx: Subcutaneous & $1182$ & $1188$ & $1333$ & $1333$ & $1185$ & $\mathbf{1469}$&$1302$ & $1188$& $1279$\\
      GTEx: Omentum & $549$ & $553$ & $1037$ & $724$ & $558$ & $\mathbf{1360}$&$635$ & $567$&$707$\\
      RNA-Seq: Bottomly \tnote{*} &$1583$ &$1693$ &$2109$ &$1714$ &$\mathbf{2347}$ &$2144$ & $2151$& $ 2167$&2142\\
      RNA-Seq: Pasilla  & $687$& $687$&$853$ &$785$ &$740$ &$856$ & $\mathbf{857}$ &$ 708$ &$\mathbf{879}$\\
      RNA-Seq: Airway \tnote{*}
        & $4079$&$4079$ &$6045$ &$4862$ &$4792$ &$\mathbf{6050}$ & $5755$&$ 5747$ &$5731$\\
      Microbiome: enigma\_ph & $61$& $65$&$96$ &$90$ &$104$ &$124$ & $126$ & $ 110$ &$\mathbf{169}$\\
      Microbiome: enigma\_al & $206$&$437$ &$496$ &$283$ &$460$ &$\mathbf{503}$ & $517$ & $ 470$ &$478$\\
      Proteomics & $244$&$358$ &$384$ &$245$ &$406$ &$\mathbf{409}$ & $394$ & $381$ & $387$\\
      fMRI: Auditory & $888$ & $888$ & - & $1015$ & $889$&$1045$ &$\mathbf{1047}$ & $956$ &$\mathbf{1125}$\\
      fMRI: Imagination & $2141$ & $2228$ & - & $2151$ & $2143$&$2237$ &$\mathbf{3055}$ & $2938$&$\mathbf{2982}$\\
      \bottomrule
    \end{tabular}
    \begin{tablenotes}
    \item[*] In the bottomly and airway data sets, the $p$-value distribution has various spikes, specifically one close to $0.9$. This null $p$-value distribution violates the (super)-uniform assumption, therefore we reselect random $p$-values in the blue region from a $\mathrm{Unif}[\lambda,\nu]$ distribution.
    \end{tablenotes}
    \end{threeparttable}
  \end{table}
 }

  In Table~\ref{table: AdaFDR Experiments}, we recreate the case study experiments in \cite{AdaFDR}. We see relatively similar performance between \gmm, \adaptg, AdaPT, and AdaFDR. However, AdaFDR only provides asymptotic FDP control, and in our logistic experiments, we found that AdaFDR may violate FDR control for larger values of $\alpha$.

\section{Supplementary Figures} \label{app: full case studies}

Figure~\ref{fig: logistic p-value dist} is comprised of the $p$-value distributions for the one-sided, two-sided, and interval null hypotheses in the logistic simulation.

\begin{figure}[t]
\centering
\begin{subfigure}{.33\textwidth}
  \centering
\includegraphics[width=0.98\columnwidth]{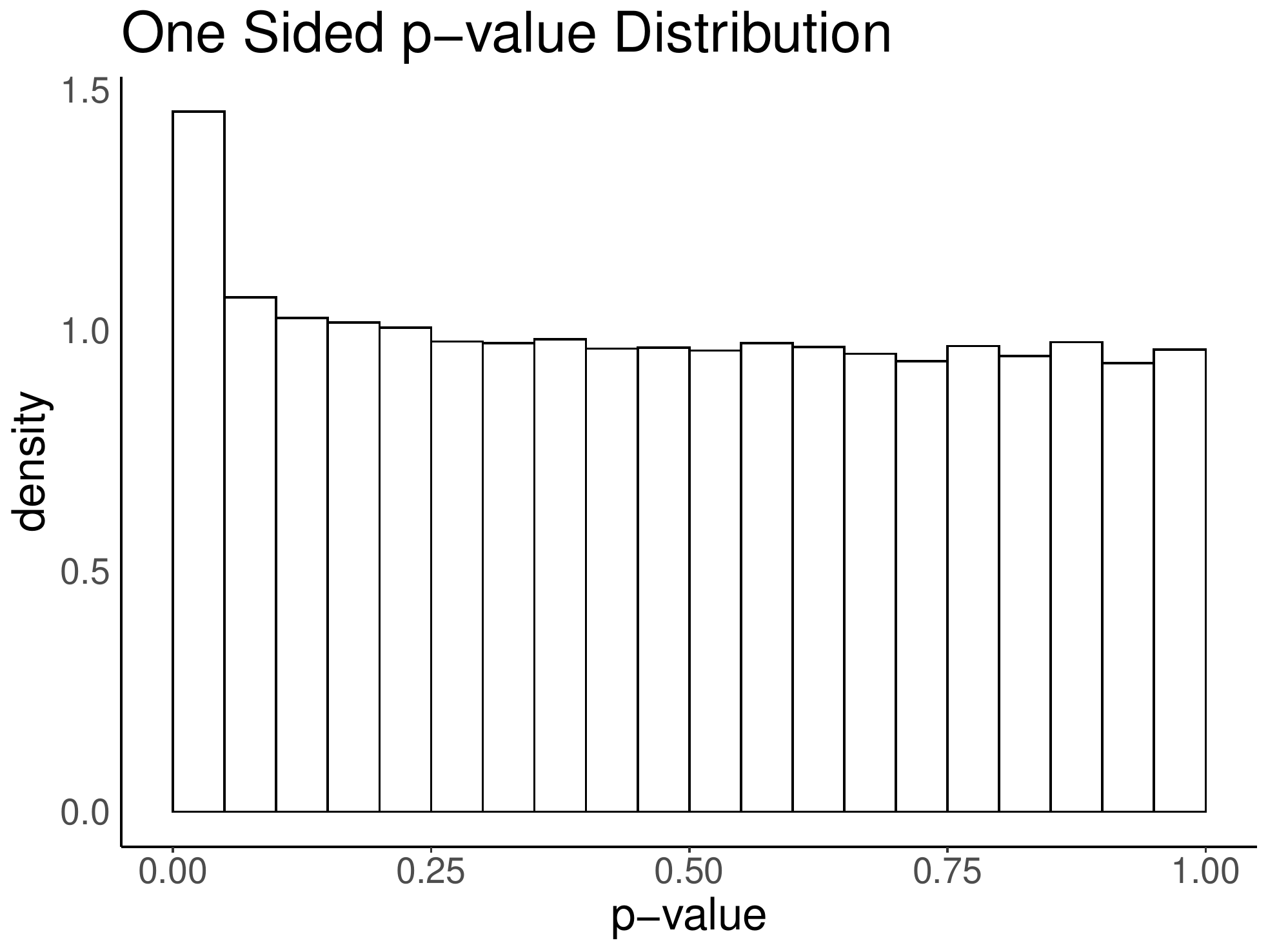}
\end{subfigure}\hfill
\begin{subfigure}{.33\textwidth}
  \centering
\includegraphics[width=0.98\columnwidth]{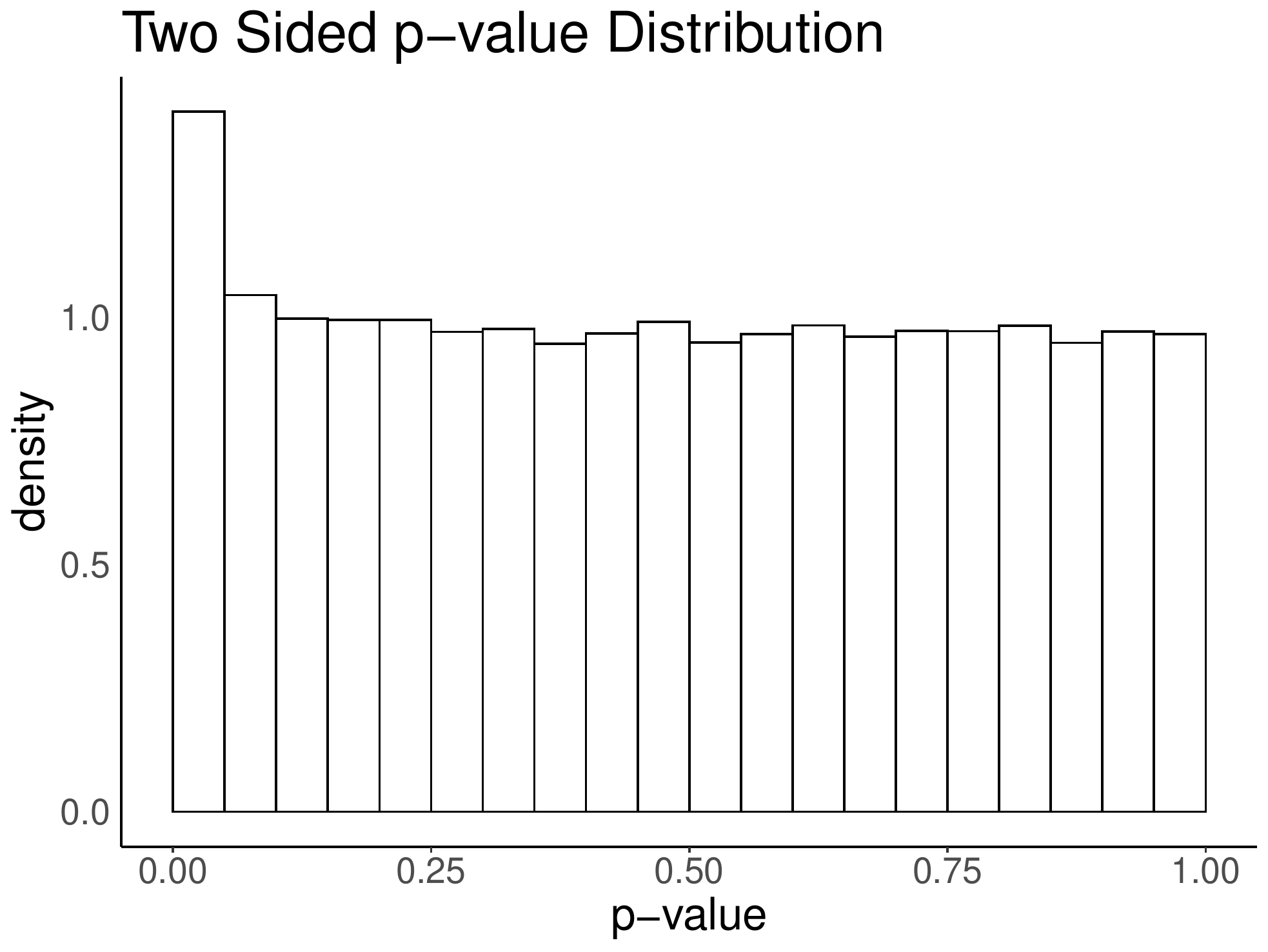}
\end{subfigure}\hfill
\begin{subfigure}{.33\textwidth}
  \centering
  \includegraphics[width=0.98\columnwidth]{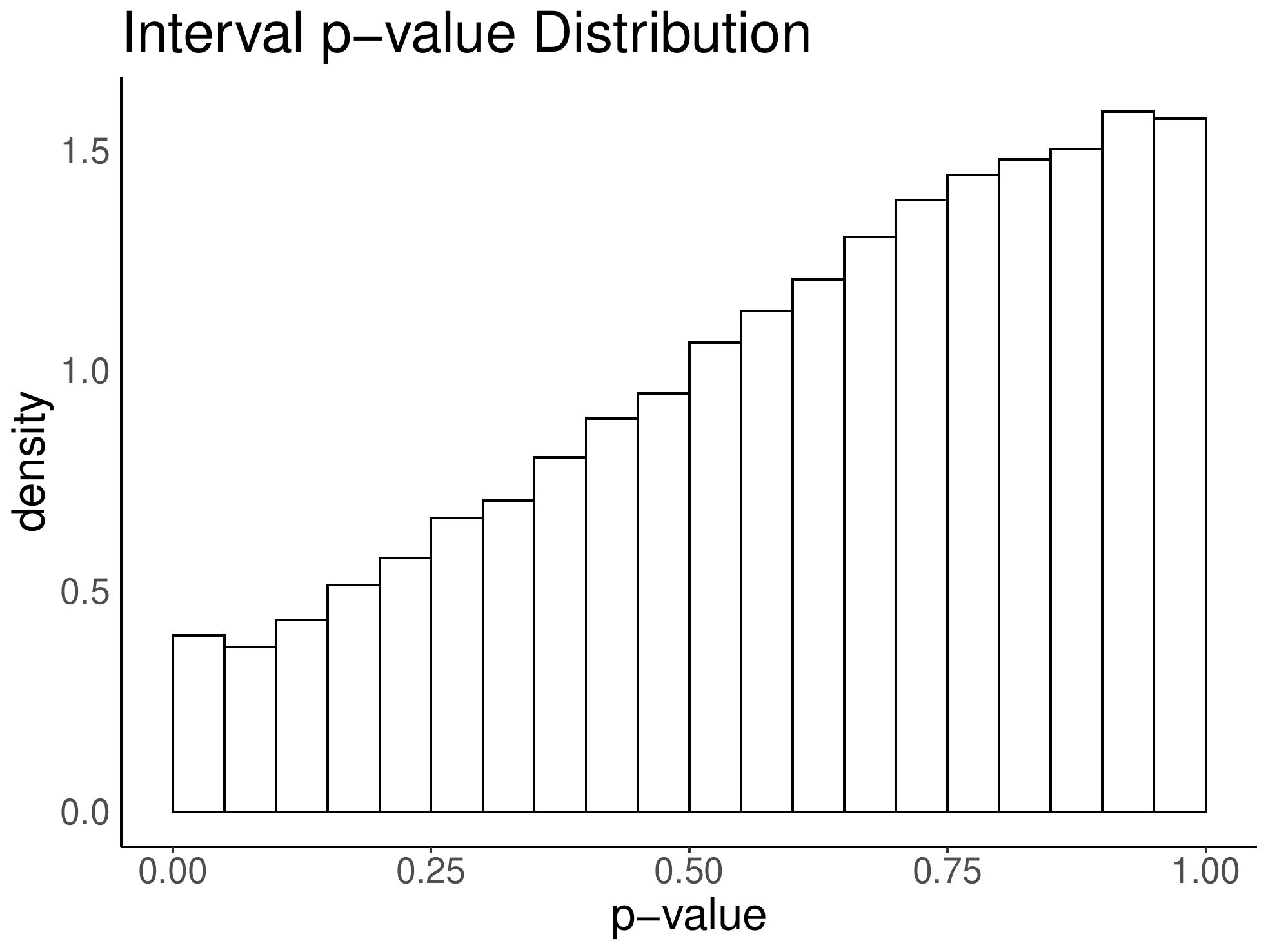}
\end{subfigure}
\caption{Distribution of $p$-values under one-sided, two-sided, and interval null testing for our logistic simulations.}
\label{fig: logistic p-value dist}
\end{figure}

\subsubsection{FDR Simulations}\label{subsec: FDR Simulations}
\begin{figure}[h]
\centering
  \includegraphics[width=\columnwidth]{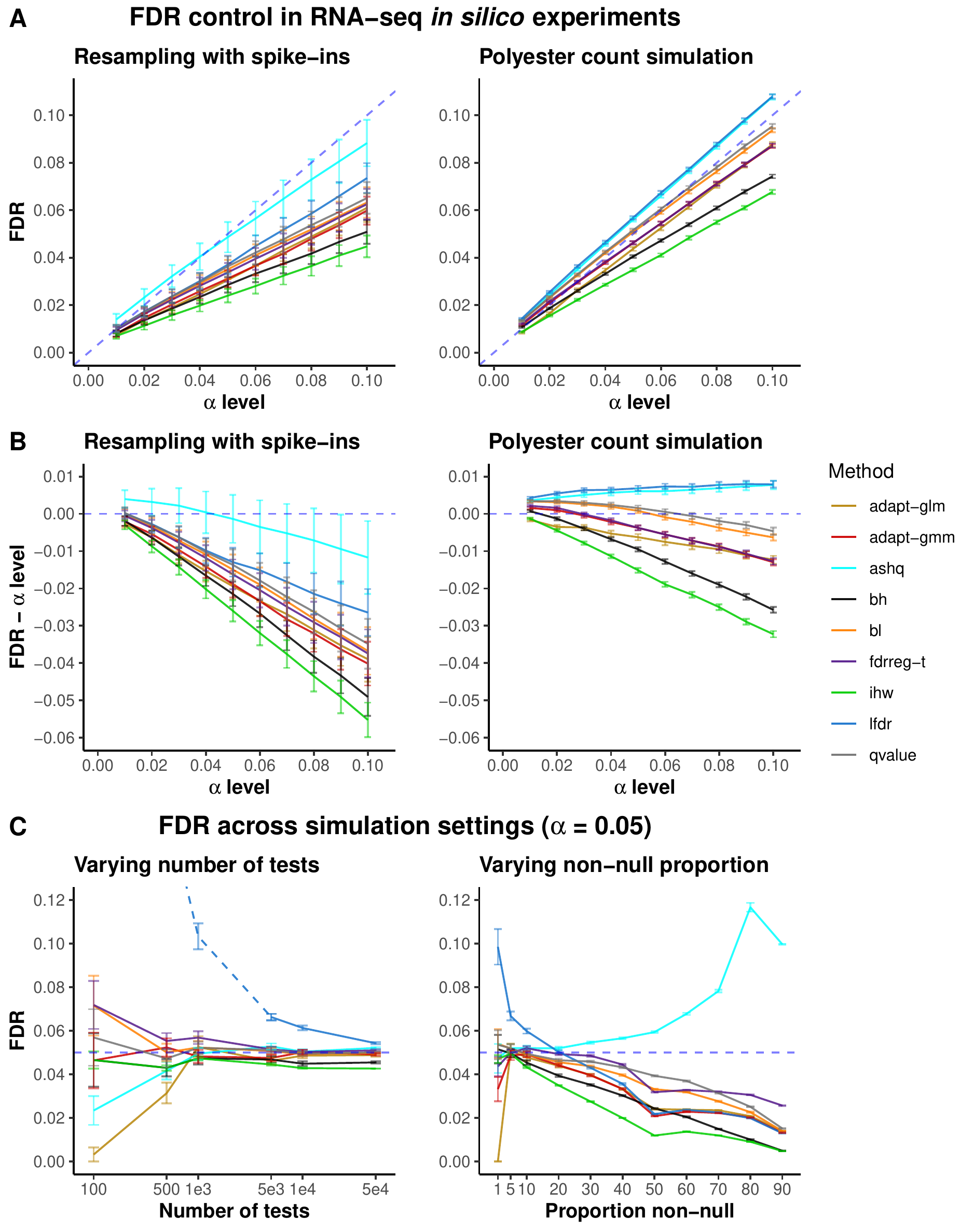}
\caption{Mean FDR and standard error over $100$ replications. $\textbf{A:}$ FDR level for various alpha levels for differential expression experiments.  $\textbf{B:}$ FDR - $\alpha$ level for greater clarity. Values above the dotted line represent violations of FDR control. $\textbf{C:}$ FDR at $\alpha=0.05$ across various number of tests and proportion of non-null hypotheses. The LFDR method is a dotted line for situations with the number of tests per bin is less than $200$ due to a generated warning.}
\label{fig: FDR experiments}
\end{figure}

\cite{Korthauer} perform simulations using yeast in silico spike-in data sets as well as simulated RNA-seq data from the \texttt{polyester} R package \citep{Polyester}. For the yeast in silico experiments in Figure~\ref{fig: FDR experiments}, $30\%$ of the genes are differentially expressed, belonging to the alternative, with a strongly informative covariate.

In Figures~\ref{fig: FDR experiments}A and B, we plot the FDR and FDR minus the $\alpha$ level, averaged over 100 replications. Points above the dotted blue line at $0$ represent violations of FDR control, while points below the line represent conservative procedures. In Figure~\ref{fig: FDR experiments}C, we analyze how the FDR differs with respect to the number of tests and the proportion of non-null hypotheses.

\gmm, in red, performs very similar to AdaPT in terms of FDR control, consistently below the desired $\alpha$ level. We find that ASH, LFDR, FDRreg violate FDR control in a various situations, in particular in small sample situations and extreme proportions of non-null hypotheses. We choose to exclude ASH, LFDR, and FDRreg methods in future case study experiments due to inconsistent FDR control.

\subsubsection{TPR Simulations}
\begin{figure}
\centering
  \includegraphics[width=\columnwidth]{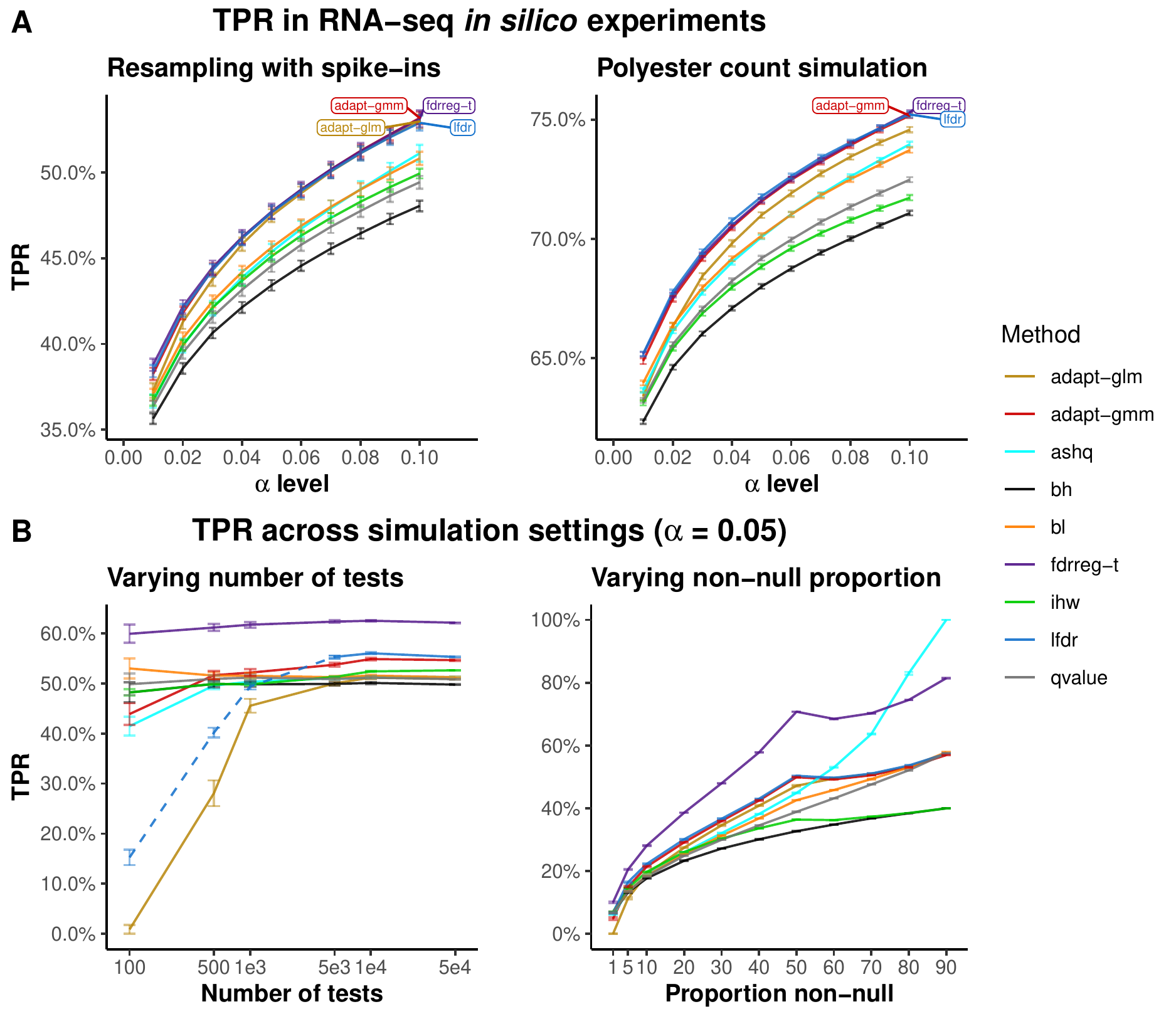}
\caption{Mean TPR and standard error over $100$ replications.  $\textbf{A:}$ TPR for various $\alpha$ levels for differential expression experiments. $\textbf{B:}$ TPR at $\alpha=0.05$ across various number of tests and proportion of non-null hypotheses. The LFDR method is a dotted line for situations with the number of tests per bin is less than $200$ due to a generated warning.}
\label{fig: TPR experiments}
\end{figure}
\cite{Korthauer} also compare power in their yeast in silico and RNA-seq polyester simulations, summarized in Figure~\ref{fig: TPR experiments}. The true positive rate (TPR), the proportion of non-null hypotheses that are rejected, is plotted on the y-axis. All methods provide improvements ($\approx$ 5-10\%) in TPR compared to the baselines of BH and Storey's q-values.

In Figure~\ref{fig: TPR experiments} we see that AdaPT, \gmm, FDRreg, and LFDR achieve the greatest TPR. However, as mentioned in the FDR experiments in Figure~\ref{fig: FDR experiments}, FDRreg, LFDR, and ASH are shown to violate FDR control.

In the left panel of Figure~\ref{fig: TPR experiments}B, we begin to see substantial differences between AdaPT and \gmm. In low sample regimes, AdaPT suffers from very low power, an example of the shortcoming from Section~\ref{subsec: AdaPT Shortcomings}. With the adaptive masking function of \gmm, we maintain a competitive TPR even for a small number of tests. 

Overall, we see that \gmm, FDRreg, and LFDR achieve the greatest TPR, and  \gmm consistently outperforms AdaPT with the new masking function and Gaussian mixture model.

\subsection{Summary Metrics}\label{subsec: Summary Metrics}
\begin{figure}
\centering
  \includegraphics[width=\columnwidth]{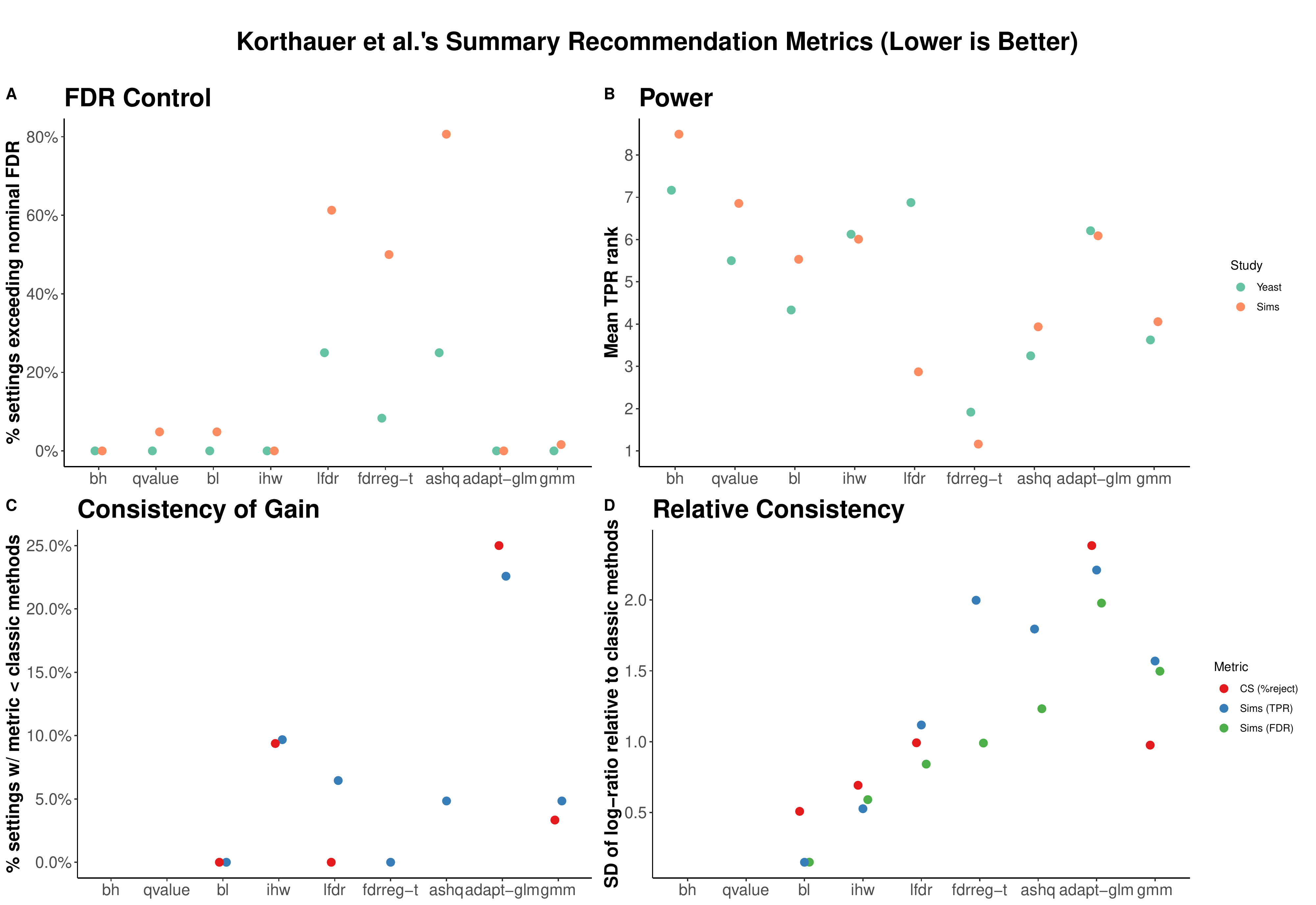}
\caption{Summary metrics for FDR control, power, and consistency for each method, lower values are more desirable for all plots (CS = case studies, Sims = simulations). \textbf{A:} The percentage of settings where FDR control is violated. \textbf{B:} The power ranking relative to the other methods. \textbf{C:} Percentage of settings where the method performs worse than the baselines of BH and Storey's q-value. \textbf{D:} The standard deviation of the log of power between the method and the better baseline between BH and Storey's q-value.}
\label{fig: Summary Metrics}
\end{figure}
To summarize all of the results from the simulations and case studies, \cite{Korthauer} construct aggregate metrics for FDR, TPR, and consistency relative to the classical methods. In all four plots, lower values are more desirable.

Figure~\ref{fig: Summary Metrics}A is a plot of the percentage of situations where the true FDP exceeds the nominal FDR. We see explicitly that LFDR, FDRreg, and ASH often violate FDR control, up to $80\%$ of the time, while AdaPT and \gmm rarely do so. 

Figure~\ref{fig: Summary Metrics}B is a plot of the mean TPR rank, aggregated over the $9$ methods, where lower ranks imply more powerful methods. We see that \gmm is more powerful than AdaPT, and is the most powerful among methods that control FDR.

In figures~\ref{fig: Summary Metrics}C and D, \cite{Korthauer} evaluate performance over various case studies (CS) and simulations (sims). Figure~\ref{fig: Summary Metrics}C is a plot of the percentage of situations where the TPR/rejection percentage is inferior to baseline methods, BH and Storey's q-value. Figure~\ref{fig: Summary Metrics}D is a plot of the standard deviation of the log ratio of power between the modern method and baseline. In these two figures, lower values imply greater consistency and improvements relative to the baseline methods. We note that AdaPT suffers as it may underperform BH and Storey's q-value if the number of tests is small or if an intercept only model is omitted. We observe that \gmm is more consistent than AdaPT, as it rarely underperforms the classic methods (0-5\% of the time) and has lower variation in performance.

\begin{landscape}
\begin{figure}
\centering
  \includegraphics[width=\columnwidth]{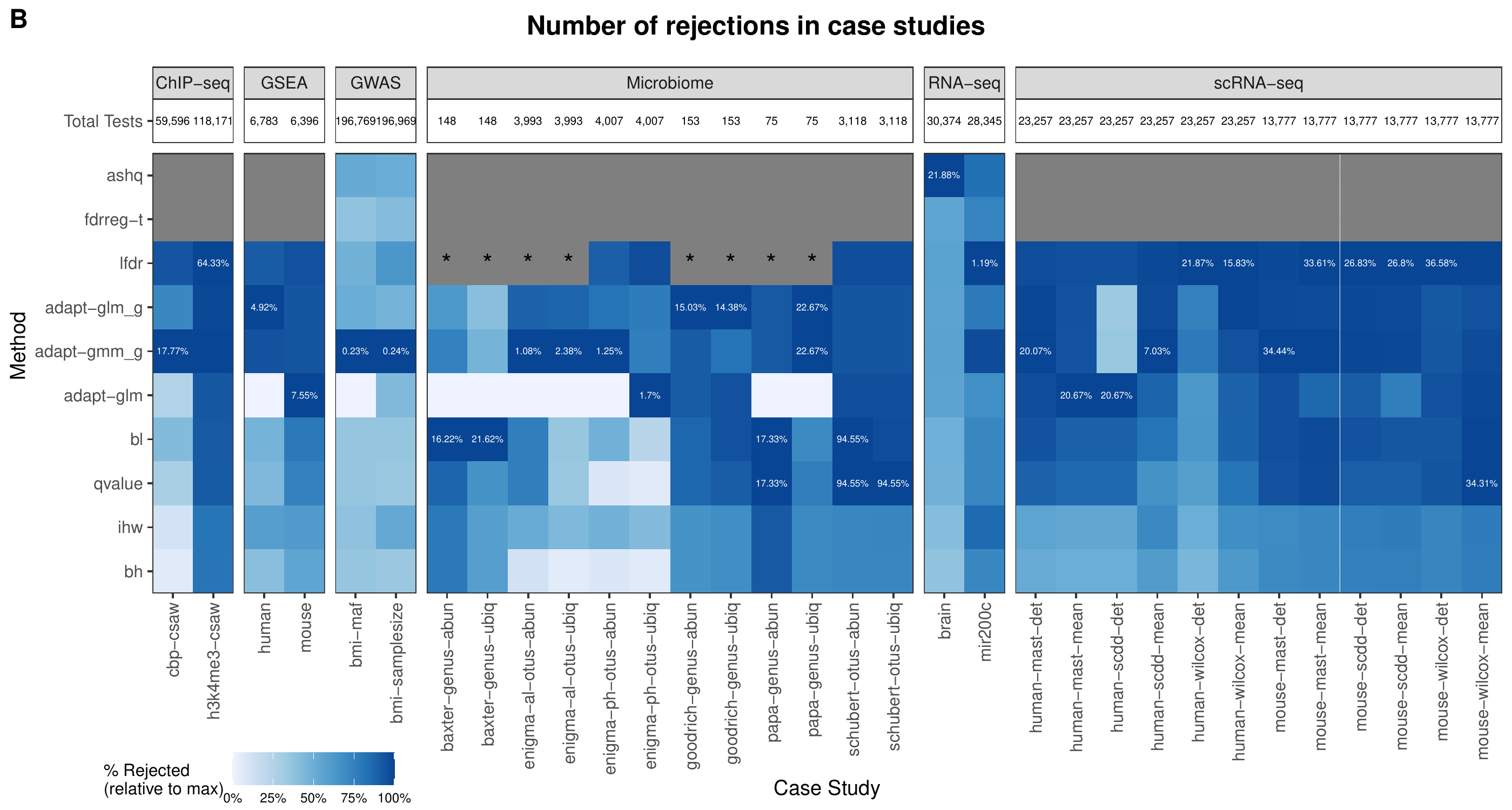}
  \caption{ Table of power in various case studies. Each row corresponds to an FDR procedure and each column corresponds to a case study and covariate. Darker squares correspond to more rejections and squares with white text correspond to the maximum rejection percentage within the column. For ASH and FDRreg, they cannot be applied in many circumstances due constraints on input, for example ASH requires the effect size and standard error. The squares with asterisks correspond to situations where LFDR is not applicable due to small sample sizes. We see that LFDR is the most powerful in many situations and \gmm provides a competitive alternative, however we stress that ASH, FDRreg, and LFDR do not control finite sample FDR as shown in Figures~\ref{fig: FDR experiments} and~\ref{fig: Summary Metrics}.
 }
\label{fig: Case Study experiments Full}
\end{figure}
\end{landscape}

\end{document}